\documentclass[12pt]{article}
\usepackage{mathrsfs}
\usepackage{amssymb}
\usepackage{amsmath}
\usepackage{ascmac}
\usepackage{amsthm}
\usepackage[pdftex]{graphicx}
\usepackage{natbib}
\usepackage{setspace}
\usepackage{color}
\usepackage{url}
\usepackage{booktabs}

\usepackage{comment}
\usepackage{bm}
\usepackage{bbm}

\usepackage{titlesec}
\titleformat*{\section}{\large\bfseries}
\titleformat*{\subsection}{\it}
\titleformat*{\subsubsection}{\it}

\newtheorem{thm}{Theorem}
\newtheorem{lem}{Lemma}

\newtheorem{prp}{Proposition}

\newtheorem{algo}{Algorithm}

\def\L{{{\cal L}}}

\def\al{{\alpha}}

\def\ga{{\gamma}}
\def\de{{\delta}}

\def\la{{\lambda}}

\def\th{{\theta}}

\def\bde{{\text{\boldmath $\delta$}}}

\def\bsi{{\text{\boldmath $\sigma$}}}
\def\bom{{\text{\boldmath $\omega$}}}

\def\bxi{{\text{\boldmath $\xi$}}}

\def\alt{{\tilde \al}}

\def\Th{{\Theta}}
\def\De{{\Delta}}
\def\Si{{\Sigma}}
\def\Ga{{\Gamma}}
\def\Om{{\Omega}}

\def\La{{\Lambda}}

\def\bTh{{\text{\boldmath $\Th$}}}
\def\bDe{{\text{\boldmath $\De$}}}
\def\bSi{{\text{\boldmath $\Si$}}}
\def\bGa{{\text{\boldmath $\Ga$}}}
\def\bOm{{\text{\boldmath $\Om$}}}
\def\bLa{{\text{\boldmath $\La$}}}

\def\bPsi{{\text{\boldmath $\Psi$}}}
\def\bPhi{{\text{\boldmath $\Phi$}}}

\def\bSih{{\widehat \bSi}}

\def\bSit{{\widetilde \bSi}}
\def\bPsit{{\widetilde \bPsi}}

\def\0{{\text{\boldmath $0$}}}
\def\a{{\text{\boldmath $a$}}}
\def\b{{\text{\boldmath $b$}}}
\def\c{{\text{\boldmath $c$}}}
\def\d{{\text{\boldmath $d$}}}
\def\e{{\text{\boldmath $e$}}}

\def\n{{\text{\boldmath $n$}}}

\def\r{{\text{\boldmath $r$}}}

\def\A{{\text{\boldmath $A$}}}
\def\B{{\text{\boldmath $B$}}}
\def\C{{\text{\boldmath $C$}}}
\def\D{{\text{\boldmath $D$}}}
\def\E{{\text{\boldmath $E$}}}

\def\I{{\text{\boldmath $I$}}}

\def\M{{\text{\boldmath $M$}}}
\def\N{{\text{\boldmath $N$}}}
\def\O{{\text{\boldmath $O$}}}
\def\P{{\text{\boldmath $P$}}}
\def\Q{{\text{\boldmath $Q$}}}
\def\R{{\text{\boldmath $R$}}}
\def\S{{\text{\boldmath $S$}}}

\def\U{{\text{\boldmath $U$}}}
\def\V{{\text{\boldmath $V$}}}
\def\W{{\text{\boldmath $W$}}}
\def\X{{\text{\boldmath $X$}}}
\def\Y{{\text{\boldmath $Y$}}}

\def\bbh{{\hat \b}}

\def\dbt{{\widetilde \d}}
\def\Dbt{{\widetilde \D}}

\def\Bbt{{\widetilde \B}}

\def\tr{{\rm tr\,}}
\def\diag{{\rm diag\,}}

\def\etr{{\rm etr\,}}

\def\1r{{\rm (1)}}
\def\2r{{\rm (2)}}
\def\3r{{\rm (3)}}
\def\4r{{\rm (4)}}
\def\5r{{\rm (5)}}

\def\non{{\nonumber}}
\DeclareMathOperator*{\bdiag}{{\bf diag\,}}
\begin{document}
\title{Gibbs Sampler for Matrix Generalized Inverse Gaussian Distributions}
\author{
Yasuyuki Hamura\footnote{Corresponding author. 
Graduate School of Economics, Kyoto University, Yoshida-Honmachi, Sakyo-ku, Kyoto, 606-8501, JAPAN. 
\newline{
E-Mail: yasu.stat@gmail.com}}, \
Kaoru Irie\footnote{Faculty of Economics, The University of Tokyo. 
\newline{
E-Mail: irie@e.u-tokyo.ac.jp}}, \
and 
Shonosuke Sugasawa\footnote{Center for Spatial Information Science, The University of Tokyo. 
\newline{
E-Mail: sugasawa@csis.u-tokyo.ac.jp}} 
}
\maketitle
\begin{abstract}
Sampling from matrix generalized inverse Gaussian (MGIG) distributions is required in Markov Chain Monte Carlo (MCMC) algorithms for a variety of statistical models. However, an efficient sampling scheme for the MGIG distributions has not been fully developed. 
We here propose a novel blocked Gibbs sampler for the MGIG distributions, based on the Choleski decomposition.
We show that the full conditionals of the diagonal and unit lower-triangular entries are univariate generalized inverse Gaussian and multivariate normal distributions, respectively. 
Several variants of the Metropolis-Hastings algorithm can also be considered for this problem, but we mathematically prove that the average acceptance rates become extremely low in particular scenarios. 
We demonstrate the computational efficiency of the proposed Gibbs sampler through simulation studies and data analysis.

\par\vspace{4mm}
\end{abstract}

\section{Introduction}
\label{sec:introduction}

The Matrix generalized inverse Gaussian (MGIG) distribution is a probability distribution for a positive definite matrix, whose probability density function at $p\times p$ matrix $\bSi$ is proportional to $|\bSi|^{\lambda} \etr ( -(\bPsi\bSi + \bGa\bSi^{-1})/2 )$ with real $\lambda$ and positive definite $\bPsi$ and $\bGa$. As a multivariate extension of the generalized inverse Gaussian (GIG) distribution, the MGIG distribution frequently appears in many statistical models and computations, including Bayesian principal component analysis and partial Gaussian graphical models. However, to the best of our knowledge, no methodology for the direct simulation from the MGIG distribution has been known (see, for example, \citealt{obiang2022bayesian}, Remark~5.1), except for restricted cases where either $\bPsi$ or $\bGa$ has rank $1$ \citep{fang2020bayesian}. 
The Markov chain Monte Carlo (MCMC) methods for the MGIG distributions has not been fully investigated either; The only exception is the hit-and-run Metropolis Hastings (MH) method proposed in \cite{yx2017}. 
Although several methods for importance sampling have been proposed (\citealt{yoshii2013infinite,yang2013multi,fazayeli2016matrix}), they are not directly applicable to the full posterior inference for the MGIG distribution. 

The objective of our study is to propose a new MCMC sampler for the MGIG distribution, evaluate its efficiency and illustrate its computational performance in applications. Specifically, we find a Gibbs sampler available for the MGIG distribution and useful in posterior inference. In constructing the Gibbs sampler, we explicitly derive the conditional distributions of the components of the MGIG-distributed matrix, utilizing its Choleski decomposition, similarly to the Bartlett decomposition of the Wishart distribution. The resulting diagonal matrix and unit lower-triangular matrix are not independent, but their conditional distributions become the univariate GIG distributions and multivariate normal distributions, and a Gibbs sampler is naturally obtained as the iterative sampling from those distributions. Our proposed Gibbs sampler is efficient in terms of effective sample size, at a small cost of increased computational time, as demonstrated in the numerical study. 

One might think that the idea of importance sampling in the literature can be imported to the independent MH methods and can construct samplers that are easier and faster than the Gibbs sampler we propose. To clarify the advantage of the Gibbs sampler, we also study the possible independent MH methods, where we use the Wishart distribution to approximate the MGIG distribution as a proposal distribution, following the comments made in the Supplementary Materials of \cite{yoshii2013infinite}. 
As reported in the literature, this approximation is reasonably well in some cases, especially when order $\lambda$ is sufficiently large, while being simple and fast in the implementation of the MH algorithms. However, we found that for certain choices of parameters of the MGIG distribution, $(\lambda, \bPsi, \bGa)$, the Wishart proposal distribution suffers from poor accuracy of the approximation, resulting in an extremely low acceptance rate. We support this claim by analytically evaluating the average acceptance rate of the MH method, as well as comparing it with the Gibbs sampler in the numerical experiments. 

The rest of this paper is organized as follows. We review the basic property of the MGIG distributions and introduce three MH methods in Section~\ref{sec:mh}, discussing that the average acceptance rate of the MH method can be extremely low in particular cases. In Section~\ref{sec:gibbs}, we compute the conditional distributions of the MGIG distributed matrix, deriving the Gibbs sampler we recommend. In Section~\ref{sec:num}, we illustrate the MH and Gibbs samplers in simulation studies and real data analysis. Examples used in this section include the MCMC analysis of the MGIG distribution, the posterior inference for the partial graphical Gaussian models, and the development of matrix-variate skewed-$t$ distributions. 
R code implementing the proposed sampler is available at GitHub repository (\url{https://github.com/sshonosuke/MGIG}).

\bigskip 
\textbf{Notations:} Unless specified, all the matrices are $p\times p$ and in bold type. Let $\O $ and $\I $ be the zero matrix and the identity matrix, respectively. For $i = 1, \dots ,p $, let $\e _i$ denote the $p$-dimensional unit vector; the $i$-th entry is unity and the others are zeros. For matrix $\C$, its sub-matrix is written as $( \C )_{\underline{i}:{\overline{i}}, \underline{j}:{\overline{j}}} =(\C_{i,j})^{i=\underline{i},\dots, \overline{i}}_{j=\underline{j},\dots, \overline{j}}$ for $1 \le \underline{i} \le \overline{i} \le p$ and $1 \le \underline{j} \le \overline{j} \le p$. For $\A = \bdiag ( a_1 , \dots , a_m )$ with positive diagonals, we write $\A ^{1 / 2} = \bdiag ( \sqrt{a_1} , \dots , \sqrt{a_p} )$ and $\A ^{- 1 / 2} = \bdiag ( 1 / \sqrt{a_1} , \dots , 1 / \sqrt{a_p} )$. 

\section{Failure of Metropolis-Hastings methods} 
\label{sec:mh}

\subsection{MGIG and Wishart distributions}

The matrix generalized inverse Gaussian distributions, denoted by ${\rm{MGIG}}_p ( \la , \bPsi , \bGa )$ with real valued $\lambda$ and positive definite matrices $\bPsi$ and $\bGa$, have the following density function \citep{barndorff1982exponential}:
\begin{equation*}
    {\rm{MGIG}}_p (\bSi | \lambda, \bPsi, \bGa) = c_p(\lambda ,\bPsi,\bGa )^{-1} | \bSi |^{\la } \exp \{ - \tr ( \bPsi \bSi + \bGa \bSi ^{- 1} ) / 2 \},
\end{equation*}
where the normalizing constant is explicitly given as 
\begin{equation*}
    c_p(\lambda ,\bPsi,\bGa ) = 2^{-\lambda _0 p} | \bGa | ^{\lambda _0} \mathcal{B}_{\lambda _0} (\bPsi \bGa / 4) , \ \ \ \ \ \mathrm{where} \ \lambda = \lambda _0 - \frac{p+1}{2}
\end{equation*}
and $\mathcal{B}_{\lambda _0} (\cdot)$ is the matrix-augment modified Bessel function of the second kind \citep{herz1955bessel}. If $\bSi \sim {\rm{MGIG}}_p ( \la , \bPsi , \bGa )$, then $\bSi ^{-1}\sim {\rm{MGIG}}_p ( -\la -(p+1), \bGa , \bPsi )$, so we assume $\lambda > -(p+1)/2$ without loss of generality. Also, the MGIG distributions with rank-deficient matrix parameters are well-defined. Specifically, the following cases are allowed: 
\begin{itemize}
    \item $\bPsi$ is positive definite, $\bGa$ is non-negative definite, and $\lambda > -1/2$, or %
    \item $\bPsi$ is non-negative definite, $\bGa$ is positive definite, and $\lambda < -p$. %
\end{itemize}
See, for example, \cite{butler1998generalized}. When either of the matrix parameters is rank-deficient, one can utilize the Matsumoto-Yor property and reduce the problem to the sampling from the MGIG distribution with lower-dimensional but full-rank matrix parameters. For this reason, we can also assume that both $\bPsi$ and $\bGa$ are positive definite. For details, see Appendix~\ref{app:deg}. Finally, re-scaled $\bSi$ also follows the MGIG distributions: if $\bSi ^{\ast} = \C\bSi \C^{\top}$ for some full-rank matrix $\C$, then $\bSi^{\ast} \sim {\rm{MGIG}}_p ( \la ,  (\C^{\top})^{-1} \bPsi \C^{-1}, \C\bGa \C^{\top} )$. For this reason, we set $\bGa = \I$ and $\bPsi$ to be diagonal in our simulation studies in Section~\ref{sec:sim}, but our method is developed for any positive definite $\bPsi$ and $\bGa$. 

To the best of our knowledge, no methodology for the direct simulation from the MGIG distribution has been known (see, for example, \citealt{obiang2022bayesian}, Remark~5.1), except for restricted cases where either $\bPsi$ or $\bGa$ has rank $1$ \citep{fang2020bayesian}. The development of the direct simulation from the general MGIG distribution is hindered mainly by the matrix Bessel function in the normalizing constant, which is hard to evaluate analytically or numerically. A Laplace approximation of the matrix Bessel function has been proposed \citep{butler2003laplace} and utilized in importance sampling \citep{yoshii2013infinite}, but its accuracy is not always satisfactory \citep{yang2013multi}. In the Bayesian principal component analysis, where the MGIG distribution arises in posterior inference, several methods of importance sampling have been proposed (\citealt{yoshii2013infinite,yang2013multi,fazayeli2016matrix}), being focused on the computation of the posterior expectation of $\bSi$ (and $\bSi^{-1}$) only. However, the proposal distribution of those importance sampling methods can also be used in the MCMC methods, as we will see below. 

One of the proposal distributions we consider is the Wishart distribution. For degree-of-freedom $\nu > p-1$ and positive definite matrix $\P$, the Wishart distribution, ${\rm{W}}_p ( \nu, \P )$, has the density evaluated at positive definite $\bSi$ as,
\begin{equation*}
    {\rm{W}}_p (\bSi | \nu, \P ) = \frac{1}{  2^{\frac{k\nu}{2}} |\P|^{\frac{\nu}{2}} \Ga_p(\frac{\nu}{2}) } | \bSi |^{ \frac{\nu}{2} - \frac{p+1}{2} } \exp \{ - \tr ( \bSi \P^{-1} ) / 2 \},
\end{equation*}
where $\Ga_p(\cdot)$ is the multivariate gamma function. 

\subsection{Metropolis-Hastings methods}

We consider the Markov chain Monte Carlo methods, targeting ${\rm{MGIG}}_p ( \la , \bPsi , \bGa )$ as the stationary distribution. Among them, the Metropolis-Hastings (MH) method is useful particularly in avoiding the evaluation of the normalizing constant of the MGIG distribution. 
The Markov kernel of transitioning $\bSi_{\rm{old}}$ to $\bSi$ of the MH method is defined by the algorithm below: for some proposal distribution $q(\cdot | \bSi_{\rm{old}})$, 
\begin{itemize}
    \item Given $\bSi_{\mathrm{old}}$, generate $\bSi _{\mathrm{new}} \sim q( \bSi _{\mathrm{new}} | \bSi_{\mathrm{old}} )$.
    
    \item Set $\bSi = \bSi _{\mathrm{new}}$ with probability 
    \begin{equation*}
        \min \left\{ 1, \frac{ {\rm{MGIG}}_p ( \bSi _{\mathrm{new}} | \la , \bPsi , \bGa ) q( \bSi _{\mathrm{old}} | \bSi_{\mathrm{new}} ) }{ {\rm{MGIG}}_p ( \bSi _{\mathrm{old}} | \la , \bPsi , \bGa ) q( \bSi _{\mathrm{new}} | \bSi_{\mathrm{old}} ) }   \right\} .
    \end{equation*}
    Otherwise, set $\bSi = \bSi _{\mathrm{old}}$. 
\end{itemize}
Note that the normalizing constant of the MGIG distribution, that involves the matrix Bessel function and is difficult to evaluate numerically, is canceled out in the acceptance rate above. 
To implement the MH method, it is necessary to construct the proposal distribution, $q( \bSi | \bSi_{\rm{old}} )$, from which it is easy to simulate. 

\subsubsection*{Independent MH method (MH1)} 

We consider a Wishart distribution whose density resembles the MGIG density as the proposal distribution of the MH method. 
This approach is classified as the independent MH method: $q(\bSi | \bSi^{\mathrm{old}}) = q(\bSi)$. Hence, the efficiency of the MCMC algorithm depends on how accurate the proposal, $q( \bSi )$, approximates the target, ${\rm{MGIG}}_p ( \bsi | \la , \bPsi , \bGa )$. By ignoring $\exp \{ - \tr (\bGa \bSi^{-1} ) /2 \}$ in the target MGIG density, \cite{yoshii2013infinite} and \cite{yang2013multi} read off the following Wishart proposal distribution: 
\begin{equation*}
    q(\bSi) = {\rm{W}}_p ( \bSi | 2 \lambda + (p+1) , \bPsi ^{-1}). 
\end{equation*}
Note that the degree-of-freedom of the Wishart distribution for a positive definite random matrix must be larger than $p-1$, so we must additionally assume $\lambda > -1$. By using this proposal distribution, the acceptance rate is, 
    \begin{equation} \label{eq:accept}
       \min \left[ 1, \exp \left\{ - \tr  \bGa ( \bSi_{\rm{new}} ^{-1} - \bSi_{\rm{old}} ^{-1} ) / 2 \right\} \right].
    \end{equation}
The effect of matrix parameter $\bGa$ on the computational efficiency of the MH method is clearly seen in the functional form of this acceptance rate. For example, if the scale of $\bGa$ increases, then it might inflate the difference between $\bSi _{\rm{new}}^{-1}$ and $\bSi _{\rm{old}}^{-1}$, leading to an extremely small acceptance rate. The other matrix parameter, $\bPsi$, does not 
appear in (\ref{eq:accept}), but in the proposal distribution, ${\rm{W}}_p ( 2 \lambda + (p+1) , \bPsi ^{-1})$. For $\bPsi$ with large eigenvalues, we expect that $\bSi _{\rm{new}}$ with small eigenvalues is generated, making $\exp \{ -\tr ( \bGa\bSi _{\rm{new}}^{-1} )/2 \}$ extremely small. %
We will investigate this acceptance rate further in Section~\ref{sec:theory}.

\subsubsection*{Mode-adjusted independent MH method (MH2)}

The log-density of the MGIG distribution is analytically tractable. The first order condition that defines the mode of the MGIG distribution is 
\begin{equation} \label{eq:riccati}
    2\lambda \bSi - \bSi \bPsi \bSi + \bGa = \O .
\end{equation}
\cite{fazayeli2016matrix} used a Wishart distribution as the proposal distribution, but proposed to adjust its mode to that of the MGIG distribution. Let $\bLa _0$ be the solution of equation~(\ref{eq:riccati}). Then, the proposal distribution is 
\begin{equation*}
q(\bSi) = {\rm{W}}_p ( \bSi | \rho_0 , \Lambda_0/(\rho_0-p-1) ),
\end{equation*}
where $\rho_0\ge p+1$ is a tuning parameter. %
Equation~(\ref{eq:riccati}) is an algebraic Riccati equation, and its unique solution, $\bLa_0$, can be numerically computed. In implementing this method, we utilize the CARE solver (the R-package \texttt{icare}) as practiced in the literature. The acceptance rate of this algorithm is easily computed as well.

\subsubsection*{Hit-and-run MH method (HR).} 

\cite{yx2017} apply the hit-and-run algorithm, which is originally proposed in \cite{yb1994}, to the case of the MGIG distribution. In constructing proposal distribution $q(\bSi | \bSi _{\rm{old}})$, this approach uses the additive noise to the ``log-scaled'' $\bSi_{\rm{old}}$, while restricting the newly generated $\bSi _{\rm{new}}$ to be positive definite. 

To detail the algorithm, let $\exp ( \A ) = \sum_{k = 0}^{\infty } \A ^k / (k !)$. 
For any positive definite matrix $\A$, let $\log ( \A )$ be the unique symmetric matrix such that $\exp \{ \log ( \A ) \} = \A $. 
Then, the HR algorithm is summarized as follows. Given a current value $\bSi _{\rm{old}}$, 
\begin{itemize}
\item Sample $l_{i, j}$ ($1 \le i \le j \le p$) and $v$ from ${\rm{N}} (0, 1)$ independently.

\item Set $\L$ to the symmetric matrix whose $(i,j)$-entry equals $l_{i,j}$ ($i\le j)$. 

\item Set $\bSi _{\rm{new}} = \exp \{ \log ( \bSi _{\rm{old}} ) + \V \} $, where $\V = v \L / \sqrt{\sum_{i = 1}^{p} \sum_{j = i}^{p} {l_{i, j}}^2}$. 

\item
Set $\bSi = \bSi _{\rm{new}}$ with probability 
\begin{align}
\min \left[ 1,\ \frac{ {\rm{MGIG}}_p ( \bSi _{\mathrm{new}} | \la , \bPsi , \bGa ) }{ {\rm{MGIG}}_p ( \bSi _{\mathrm{old}} | \la , \bPsi , \bGa ) }  \prod_{1 \le i < j \le p} \frac{( d_{i}^{*} - d_{j}^{*} ) ( \log d_i - \log d_j ) }{ ( \log d_{i}^{*} - \log d_{j}^{*} ) ( d_i - d_j )} \right] \text{,} \non %
\end{align}
where $d_{1}^{*} \ge \dots \ge d_{p}^{*}$ and $d_1 \ge \dots \ge d_p$ are the characteristic roots of $\bSi _{\rm{new}}$ and $\bSi _{\rm{old}}$, respectively. Otherwise, set $\bSi = \bSi _{\rm{old}}$.
\end{itemize}
\cite{yx2017} report that this MH method works reasonably well in their application, where the size of $\bSi$ is at most $p=49$ and order parameter $\lambda$ is sufficiently large. We will evaluate its empirical computational efficiency for smaller $\lambda$ in our simulation study in Section~\ref{sec:sim}. %

\subsection{Efficiency of the MH methods}
\label{sec:theory}

The efficiency of the independent MH method depends on the accuracy of the approximation of the original MGIG distribution by the Wishart distribution or, equivalently, the acceptance rate. In this subsection, we study the acceptance rate of MH1 in (\ref{eq:accept}). 

Although it is difficult to obtain the clear, interpretable bounds of the acceptance rate, we can still gain some insights on the efficiency of the MH method from simple examples by computing the average acceptance rate (AAR, \citealt[Section~7.6.1]{robert1999monte}), 
\begin{equation*}
    \mathrm{AAR}(\lambda, \bPsi,\bGa) = \mathbb{E}[ \exp \{ -\tr \bGa ( \bSi _{\rm{new}}^{-1} - \bSi _{\rm{old}}^{-1} ) /2 \}  ] = 2\mathbb{P}[ \tr \bGa \bSi_{\rm{old}}^{-1} \le \tr \bGa \bSi_{\rm{new}}^{-1} ], 
\end{equation*}
where $\bSi _{ \rm{new} }\sim {\rm{W}}_p ( 2 \lambda + p + 1 , \bPsi ^{-1})$ and $\bSi _{ \rm{old} }\sim {\rm{MGIG}}_p ( \lambda , \bPsi , \bGa)$. The expression above also implies that re-scaling of the matrix of interest does little to the improvement of the sampling efficiency. To be precise, for some full-rank $p{\times}p$ matrix $\C$, using $(\C\bSi_{\rm{old}}\C^{\top} , \C\bSi_{\rm{new}}\C^{\top})$ instead of $(\bSi_{\rm{old}} , \bSi_{\rm{new}})$ does not change the AAR. 

We consider two examples of the MGIG distributions and evaluate the limit of the AAR. The proofs of the statements below are given in the Supplementary Materials (Section~S3). 

\ 

\textbf{Example 1.}~(Large and small $\lambda$) The previous studies on the importance sampling and HR methods evaluate the computational performance of those methods for sufficiently large $\lambda$. For example, the $\lambda$ is at least $10$ in the examples of \cite{fazayeli2016matrix}. To investigate the effect of $\lambda$ on the AAR, %
first, we prove that $\mathrm{AAR}( \lambda , \bPsi , \bGa ) \to 1$ as $\lambda \to \infty$. This result supports the empirical findings in the literature. In contrast, when $\lambda \to -1$, we have $\mathrm{AAR}( \lambda , \bPsi , \bGa  ) \to 0$. That is, the smaller the $\lambda$ is, the harder the MH method accepts the newly generated value and the less efficient the sampler becomes. We will compute the effective sample size of the Gibbs and MH methods in simulation studies where $\lambda = 2$, a small value relative to those considered in the existing studies. 

\ 

\textbf{Example 2.}~(Large $\bPsi$) Suppose that $\bGa = \I $ and $\bPsi = \diag (\psi , 1, \dots , 1)$ for some large $\psi > 0$ and arbitrary $\lambda > -1$. %
Then, as $\psi \to \infty$, we have $\mathrm{AAP}( \lambda , \bPsi, \I ) \to 0$. This example implies the possible failure of the MH methods, where generated $\bSi_{\rm{new}}$ are hardly ever accepted, even in low-dimensional cases. In Section~\ref{sec:sim}, we evaluate the efficiency of the MH methods in similar scenarios, where $\bPsi$ has several large diagonals.

%

\section{Block Gibbs Sampler for MGIG Distribution}
\label{sec:gibbs} 

We propose a Gibbs sampler for the MGIG distribution by computing the full conditional distributions of the matrix entries. Specifically, we consider the Choleski decomposition of the positive definite matrix $\bSi$, deriving the conditional distributions of its diagonal distribution and unit lower-triangular matrix. This approach resembles the Bartlett decomposition of the Wishart distributed matrix, where all the entries of the decomposed matrices become mutually independent, the entries of the diagonal distribution follow the chi-squared distributions, and those of the unit lower-triangular entries follow the standard normal distribution. In contrast, in the case of the MGIG distribution, the entries of the decomposed distributions are not independent. Instead, we observe that the conditional distributions turn out to be the independent univariate GIG distributions and multivariate normal distributions, respectively. This observation directly leads to a Gibbs sampler we propose. 

Let $\bSi \sim {\rm{MGIG}}_p ( \lambda , \bPsi , \bGa)$. First, we consider a decomposition, $\bSi = \B \A \B ^{\top }$, where $\A={\rm diag}(a_1,\ldots,a_p)$  and 
\begin{align}
\B = \begin{pmatrix} 1 & 0 & \cdots & 0 & 0 \\ b_{2, 1} & 1 & \cdots & 0 & 0 \\ \vdots & \vdots & \ddots & \vdots & \vdots \\ b_{p - 1, 1} & b_{p - 1, 2} & \cdots & 1 & 0 \\ b_{p, 1} & b_{p, 2} & \cdots & b_{p, p - 1} & 1 \end{pmatrix} \text{,} \non 
\end{align}
for $( \a , \b ) = (( a_i )_{i = 1}^{p} , (( b_{i, j} )_{j = 1}^{i - 1} )_{i = 2}^{p} )$ being the unique point in $(0, \infty )^p \times \mathbb{R} ^{p (p - 1) / 2}$.
Here $\A$ is a diagonal matrix and $\B$ is a lower-triangular matrix, so that the decomposition $\bSi = \B \A \B ^{\top }$ is the Cholesky decomposition. It is immediate from the change-of-variables that the joint density of $\a $ and $\b $ is 
\begin{align}
p( \a , \b ) &\propto \Big( \prod_{i = 1}^{p} {a_i}^{\la + p - i} \Big) \exp [- \tr \{ \A ^{1 / 2} \B ^{\top } \bPsi \B \A ^{1 / 2} + \A ^{- 1 / 2} \B ^{- 1} \bGa ( \B ^{- 1} )^{\top } \A ^{- 1 / 2} \} / 2] \text{.} \non 
\end{align}
The conditional distribution of $\a $ given $\b $ can be easily read-off as 
\begin{align}
p( \a | \b ) &= \prod_{i = 1}^{p} {\rm{GIG}} ( a_i | \la + p - i + 1, ( \B ^{\top } \bPsi \B )_{i, i} , ( \B ^{- 1} \bGa ( \B ^{- 1} )^{\top } )_{i, i} ) \text{,} \non 
\end{align}
or the independent GIG distributions. Likewise, as the function of each lower-triangular entry of $\B$, the joint density is the exponentiated quadratic function, so the full conditional of each entry of $\B$ is a normal distribution. Furthermore, to derive a more efficient sampler, we work on the $i$-th column vector of $\B$ for $i=1,\dots ,p-1$, namely, 
\begin{align}
\begin{matrix} {} \\ {} \\ {} \\ i \to \\ {} \end{matrix} 
\begin{pmatrix} 0 \\ \vdots \\ 0 \\ 1 \\ \b _i \end{pmatrix} \hspace{- 0.5cm} \begin{array}{c}
\left. 
\begin{array}{c}
{} \\ {} \\ {}
\end{array}
\right\} i - 1 
\\
\left. 
\begin{array}{c}
{} 
\end{array}
\right. {} 
\\
\left. 
\begin{array}{c}
{} 
\end{array}
\right\} p - i 
\end{array} \begin{matrix} {} \\ {} \\ {} \\ {} \\ \text{,} \end{matrix} \non 
\end{align}
where $\b _i = ( b_{h, i} )_{h = i + 1}^{p}$ for $i = 1, \dots , p - 1$ when $p \ge 2$. 
Given $\a$ and $\b_{-i} = \b \setminus \b _i$, the conditional distribution of the $(p-i)$-dimensional vector $\b_i$ is, in fact, a multivariate normal distribution, whose mean and variance can be computed recursively as we move from $i=1$ to $i=p-1$. The observations we had so far are summarized as follows:

\begin{thm}\label{thm:cond}
The full conditional distribution of $\a $ is the product of $p$ independent generalized inverse Gaussian distributions. 
For all $i = 1, \dots , p - 1$, the full conditional distribution of $\b _i$ is a $(p - i)$-dimensional multivariate normal distribution. 
\end{thm}

The detailed proof is given in the Supplementary Material (Section~S1 and S2).
Based on the results of conditional distributions in Theorem~\ref{thm:cond}, we can develop the following Gibbs sampler to generate the MCMC samples of $\bSi$.

\begin{algo}[Block Gibbs sampler for MGIG distribution]
\label{algo:MGIG}
Assume that $p \ge 2$. 
Then the variables $\a $ and $\b _1 , \dots , \b _{p - 1}$ are updated in the following way: 
\begin{enumerate}
\item
Compute $\B ^{- 1}$, $\B ^{\top } \bPsi \B $, and $\B ^{- 1} \bGa ( \B ^{- 1} )^{\top }$. 
\item
Sample $\a ^{*} = ( a_{i}^{*} )_{i = 1}^{p} \sim \prod_{i = 1}^{p} {\rm{GIG}} ( \la + p - i + 1, ( \B ^{\top } \bPsi \B )_{i, i} , ( \B ^{- 1} \bGa ( \B ^{- 1} )^{\top } )_{i, i} )$ and let $\A ^{*} = \bdiag ( \a ^{*} )$. 
\item
For $i = 1, \dots , p$, let $\bbh _i = ( \overbrace{0, \dots , 0}^{i - 1} , 1, {\b _i}^{\top } )^{\top } \in \mathbb{R} ^p$, 
\begin{align}
&\B _i = \begin{pmatrix} \e _1 & \cdots & \e _{i - 1} & \bbh _i & \e _{i + 1} & \cdots & \e _p \end{pmatrix} \text{,} \quad \text{and} \quad \overline{\B } _i = 2 \I - \B _i \text{.} \non 
\end{align}
\item
Compute $\Q ^{*} = ( \B ^{- 1} )^{\top } ( \A ^{*} )^{- 1} \B ^{- 1}$. 
\item
For $i = 1$, 
\begin{itemize}
\item
Let $\M _{1}^{*} = \bPsi $, $\R _{1}^{*} = \overline{\B } _1 \B ( \A ^{*} )^{1 / 2}$, $\overline{\M } _{1}^{*} = \bGa $, and $\overline{\R } _{1}^{*} = {\B _1}^{\top } ( \B ^{- 1} )^{\top } ( \A ^{*} )^{- 1 / 2}$. 
\item
Let $\N _{1}^{*} = a_{1}^{*} ( \bPsi )_{2:p, 2:p} + ( \overline{\M } _{1}^{*} )_{1, 1} ( \Q ^{*} )_{2:p, 2:p}$. 
\item
Sample $\b _{1}^{*} \sim {\rm{N}}_{p - 1} (( \N _{1}^{*} )^{- 1} \n _{1}^{*} , ( \N _{1}^{*} )^{- 1} )$, where 
\begin{align}
\n _{1}^{*} = - ( \M _{1}^{*} )_{2:p, 1:p} \R _{1}^{*} (( \R _{1}^{*} )_{1, 1:p} )^{\top } + ( \overline{\R } _{1}^{*} )_{2:p, 1:p} ( \overline{\R } _{1}^{*} )^{\top } ( {( \overline{\M } _{1}^{*} )_{1, 1:p}} )^{\top } \text{,} \non 
\end{align}
and let $\bbh _{1}^{*} = (1, ( \b _{1}^{*} )^{\top } )^{\top }$, 
\begin{align}
\B _{1}^{*} = \begin{pmatrix} \bbh _{1}^{*} & \e _2 & \cdots & \e _p \end{pmatrix} \text{,} \quad \text{and} \quad \overline{\B } _{1}^{*} = 2 \I - \B _{1}^{*} \text{.} \non 
\end{align}
\end{itemize}
\item
If $p \ge 3$, then for $i = 2, \dots , p - 1$, 
\begin{itemize}
\item
Let $\M _{i}^{*} = ( \B _{i - 1}^{*} )^{\top } \M _{i - 1}^{*} \B _{i - 1}^{*}$, $\R _{i}^{*} = \overline{\B } _i \R _{i - 1}^{*}$, $\overline{\M } _i = \overline{\B } _{i - 1}^{*} \overline{\M } _{i - 1}^{*} ( \overline{\B } _{i - 1}^{*} )^{\top }$, and $\overline{\R } _{i} = {\B _i}^{\top } \overline{\R } _{i - 1}^{*}$. 
\item
Let $\N _{i}^{*} = a_{i}^{*} ( \bPsi ) _{(i + 1):p, (i + 1):p} + ( \overline{\M } _{i}^{*} )_{i, i} ( \Q ^{*} )_{(i + 1):p, (i + 1):p}$. 
\item
Sample $\b _{i}^{*} \sim {\rm{N}}_{p - i} (( \N _{i}^{*} )^{- 1} \n _{i}^{*} , ( \N _{i}^{*} )^{- 1} )$, where 
\begin{align}
\n _{i}^{*} &= - ( \M _{i}^{*} )_{(i + 1):p, 1:p} \R _{i}^{*} (( \R _{i}^{*} )_{i, 1:p} )^{\top } + ( \overline{\R } _{i}^{*} )_{(i + 1):p, 1:p} ( \overline{\R } _{i}^{*} )^{\top } ( {( \overline{\M } _{i}^{*} )_{i, 1:p}} )^{\top } \text{,} \non 
\end{align}
and let $\bbh _{i}^{*} = (0, \dots , 0, 1, ( \b _{i}^{*} )^{\top } )^{\top } \in \mathbb{R} ^p$, 
\begin{align}
&\B _{i}^{*} = \begin{pmatrix} \e _{1} & \cdots & \e _{i - 1} & \bbh _{i}^{*} & \e _{i + 1} & \cdots & \e _{p} \end{pmatrix} \text{,} \quad \text{and} \quad \overline{\B } _{i}^{*} = 2 \I - \B _{i}^{*} \text{.} \non 
\end{align}
\end{itemize}
\end{enumerate}
\end{algo}

\bigskip 
In sampling $\b_i$, we need to compute $\n_i^{*}$ and $\N_i^{*}$. In doing so, we have to update not all but some parts of $(\M_{i}^{*}, \overline{\M } _{i}^{*}, \R_{i}^{*}, \overline{\R } _{i}^{*})$. Such an update can be done fast, for it only requires the multiplication of $\B_i$ and $\overline{\B}_i$ to the existing $(\M _{i-1}^{*},\overline{\M } _{i-1}^{*}, \R_{i-1}^{*}, \overline{\R } _{i-1}^{*})$, which is not as costly as $O(p^3)$ and does not hinder the implementation of the algorithm. Note also that %
some of the necessary matrices, including the submatrix of $\Q^{*}$, depend only on $(\b_{i+1},\dots, \b_{p-1})$, but not on $(\b_1,\dots ,\b_{i-1})$. Thus, we do not have to update those matrices, such as $\Q^{*}$, as we sample each of $\b_i$'s, but need to compute them once before starting to sample $\b$. 

This algorithm involves multiple matrix decomposition and multiplication, so is clearly more computationally costly than the MH methods. One of the bottlenecks is the necessity of decomposing $i{\times}i$ matrix $\N_i^{*}$ for $i=1,\dots , p$ in every scan of the algorithm. Hence, in the case of extremely high-dimensional applications, the proposed algorithm might need more sophistication to be computationally feasible. Here we would like to point out that the decomposition of $p$ matrices, $\N_1^{*},\dots , \N_p^{*}$, can be parallelized; see the Supplementary Materials (Section~S4). In our numerical examples of Section~\ref{sec:num}, where the dimension is at most $p=100$, we do not need such an acceleration of the algorithm.

\section{Numerical Studies}
\label{sec:num}

\subsection{Random matrix generation}
\label{sec:sim} 

We first assess the performance of the proposed Gibbs sampler (GS) as well as the variants of the MH methods in Section~\ref{sec:mh}, to generate samples from MGIG distributions. 
For comparison, we also employ three Metropolis-Hastings algorithms, MH1, MH2 nad HR, in Section~\ref{sec:mh}. 
Regarding the tuning parameter $\rho = \rho _0 - p -1$ in the proposal Wishart distribution in MH2, we searched over $\rho=1,2,\ldots,10$ and set $\rho=5$ as the best choice maximizing the sampling efficiency under $p=5$. 
In this study, for ${\rm{MGIG}}_p(\lambda, \bPsi,\bGa)$, we set $\lambda=2$ and $\bGa=\I$, and considered three cases of $\bPsi$ given by 
$$
{\rm (I)}: \bPsi=\I , \ \ \ 
{\rm (II)}: \bPsi={\rm diag}(1,\ldots,1, 10, 50), \ \ \ 
{\rm (III)}: \bPsi={\rm diag}(1,\ldots,p).
$$
Regarding the dimension $p$, we considered $p=5, 10,\ldots,100$.
In implementing those samplers for various $p$ and three scenarios of $\bPsi$, we generated $50,000$ samples after discarding $5,000$ samples as burn-in.  
To evaluate the sampling performance, we compute effective sample sizes (ESS) of each element of $p\times p$ matrix and averaged ESS over all of the $p(p+1)/2$ elements. 

In Figure~\ref{fig:generation}, we report ESS and ESS per second of the four sampling algorithms under three scenarios of $\bPsi$. 
First, it is confirmed that the proposed Gibbs sampler has the highest raw ESS in most scenarios, regardless of $p$, being as efficient as the direct, independent sampling. 
In contrast, the MH methods do not work well in this study. 
While MH1 and MH2 provide reasonable ESS values under low or moderate $p$, their ESS rapidly decreases as $p$ increases, particularly in Scenarios (II) and (III). This observation is predicted from our analysis of the average acceptance rate in Section~\ref{sec:theory}. 
To be fair, we note that MH2 has a higher ESS per second that GS in Scenarios (I) and (II). In these cases, the MH method can run the MCMC algorithm longer in a short computational time. The GS method is useful for the MGIG distribution of $\bPsi$ with large diagonals, as seen in its higher ESS in Scenario (III).

We would like to emphasize that these scenarios assume $\lambda = 2$, which is small relative to the values used in the literature. This setting explains not only the superiority of the proposed Gibbs sampler, but also that the mode-adjustment of the MH2 method is outperformed by the naive Wishart approximation of the MH1 method in many cases. When $\lambda$ is large, the MH methods work better in terms of ESS and become more competitive, as confirmed in the literature and predicted from the result of Section~\ref{sec:theory}. We double-check this by conducting the additional simulation studies with $\lambda = 10$. See the Supplementary Materials (Section~S5) for its details.

The time-consuming but highly-efficient aspect of GS is also essential when the sampler is incorporated into a larger MCMC algorithm for more structured statistical models, as demonstrated in the subsequent subsections. 
\begin{figure}[!htb]
\centering
\includegraphics[width = \linewidth]{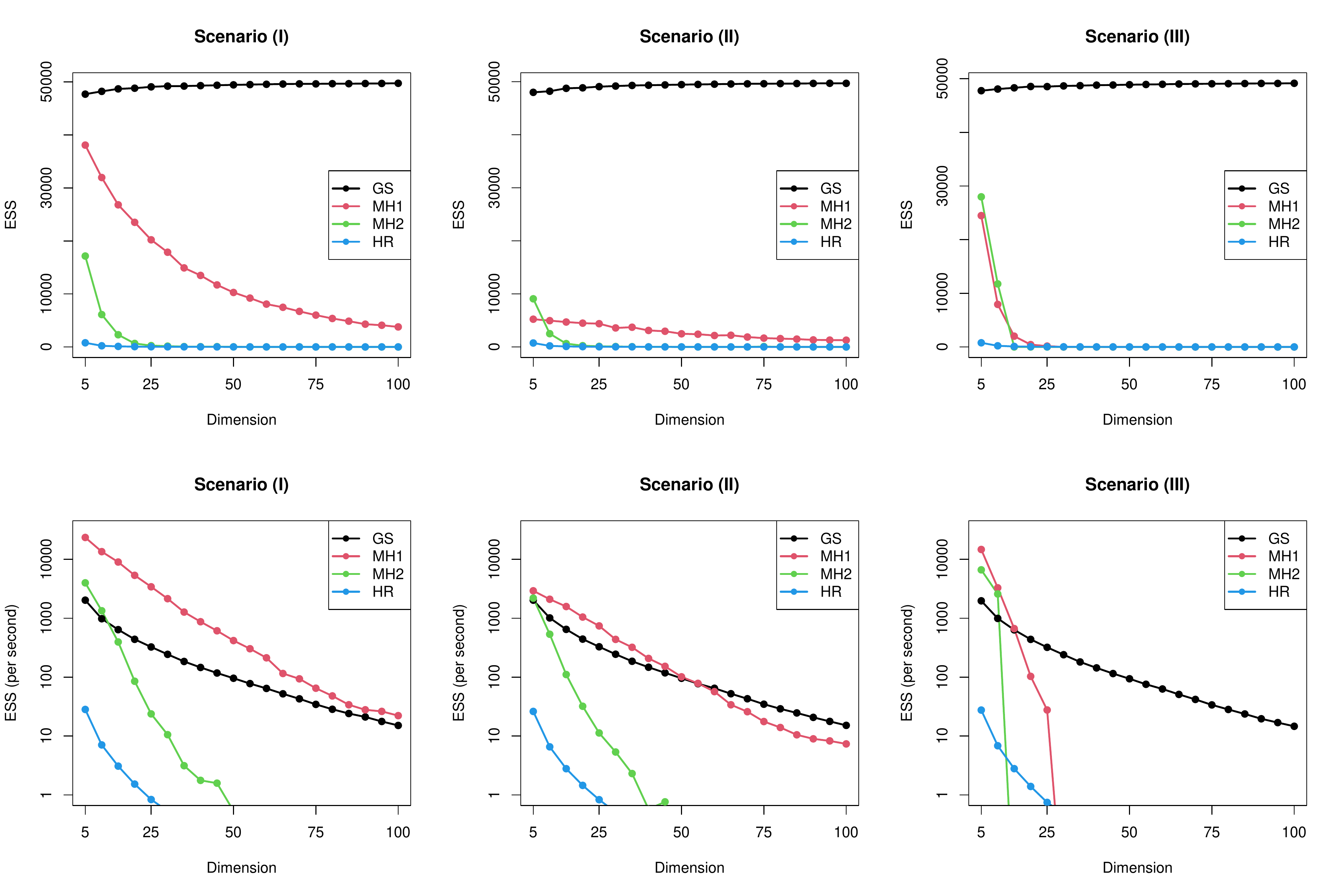}
\caption{Effective sample size (ESS) and ESS per second of the four samplers.}
\label{fig:generation}
\end{figure}

\subsection{Partial Gaussian graphical modeling}
\label{sec:pggm}

We next consider the use of the proposed Gibbs sampler as a part of MCMC algorithm. 
To this end, we here consider posterior inference on partial Gaussian graphical models. 
Let $\Y $ be an $n{\times}q$ response matrix and $\X$ an $n{\times}p$ covariate matrix. Based on Section~2 of \cite{obiang2022bayesian}, we consider the following partial Gaussian graphical model with sparsity: 
\begin{align*}
\Y|\X,\bDe,\bOm_y &\sim N_{n\times q}(-\X\bDe^{\top}\bOm_y^{-1},\I_n,\bOm_y^{-1}), \\
\bDe _k |\bOm _y, \lambda_k , \pi &\sim (1-\pi ) N_q(\0_q,\lambda_k\bOm_y) + \pi \delta_{0_q},   \ \ \ \ 
 \lambda_k \sim {\rm{Ga}}(\alpha ,\ell_k ) \ \ \ \ \ k=1,\dots ,p, 
\end{align*}
with priors $\bOm_y \sim {\rm{W}}_q( u, V )$ and $\pi \sim {\rm{Be}}( a,b )$, where $\delta_{0_q}$ is the point-mass distribution on the $q$-dimensional zero vector, $\bOm_y$ a $q{\times}q$ positive definite matrix, $\bDe$ a $q{\times}p$ regression coefficient matrix, and $\bDe_k$ the $k$th column vector of $\bDe$. 
This model can be rewritten for the conditionally-independent multivariate observations as  
\begin{equation}\label{PGGM}
\Y_i \sim N_q( \bOm_y^{-1} \bDe \X_i, \bOm_y^{-1} ), \ \ \ \ \ i=1,\dots, n, 
\end{equation}
where $\Y_i^{\top}$ and $\X_i^{\top}$ are the $i$-th row vectors of $\Y$ and $\X$, respectively. The prior for $\bDe$ is the spike-and-slab prior and introduces the sparsity in the coefficient matrix. The variance matrix, $\bOm_y^{-1}$, is also used in the location of $\Y$ to introduce the skewness of observations.

\bigskip
The MCMC algorithm for the posterior analysis of this model has been given in Proposition 2.1 of \cite{obiang2022bayesian}, except for the sampler for $\bOm _y$. 
The full conditional of $\bOm_y$ becomes the matrix generalized inverse Gaussian distribution,
\begin{equation*}
    {\rm{MGIG}_q}( (n+N_0+u-2p-1) /2, \Y^{\top}\Y+\V^{-1} , \bDe \{ \X^{\top}\X + \diag (\lambda _1^{-1},\dots, \lambda _p^{-1} \} \bDe ^{\top} ),
\end{equation*} 
where $N_0 = \sum _{i=1}^p \mathbbm{1}[\bDe _i = 0]$, the number of the all-zero column vectors of $\bDe$. 
In the original algorithm, the simulation from the MGIG distribution is replaced with the plug-in of its mode, $\bSi = \Lambda_0$, or the solution of the algebraic Riccati equation (\ref{eq:riccati}), which we call the mode imputation (MI) method in what follows. 
Hence, to be rigorous, the original algorithm in \cite{obiang2022bayesian} is not a valid MCMC method. 
Alternatively, we employ the proposed Gibbs sampler (GS), MH method with a Wishart proposal and the hit-and-run MH method (HR) to complement the original algorithm. 

\bigskip
We consider simulation studies to evaluate the performance of MCMC with various sampling (update) schemes for $\bOm_y$.
Throughout the simulation studies, we set $n=100$ and use hyperparameters $\al = (q + 1) / 2$, $l_k=1 \ (k=1,\ldots,q)$, $u = q$, $\V = \I_q / q$, and $a = b = 1$. 
Following the simulation studies in \cite{obiang2022bayesian}, we first generate each element of $\X$ from ${\rm{U}} (0, 1 / 3)$ independently, and then generate a synthetic sample $\Y_i$ from (\ref{PGGM}), where the true values are obtained as $\bOm_y=2\C_q^{-1}$ with $\C_q=(0.5^{|j-k|})_{1\leq j,k\leq q}$ and $\bDe_{k}\sim 0.5N_q(\0_q, \bOm_y)+0.5\delta_{0_q}$.

\bigskip
We first set $q=3$ and $p=10$, and run the three MCMC algorithms. In each algorithm, we obtain $20,000$ posterior samples and take samples at every five iterations after discarding the first $2,000$ samples. 
We show the traceplots for $(\bOm_y)_{11}$, $(\bOm_y)_{12}$, $\bDe_{14}$ and $\bDe_{24}$ in Figure~\ref{fig:PGGM-plot}. The efficiency of the Gibbs sampler (GS) is clear in this plot as well. The HR sampler exhibits some potential autocorrelations of the samples, implying the necessity of longer chains. The MH method is unable to sample $\bOm_y$ at all, fixing it to several values in essence. This undesirable aspect of the mixing of $\bOm_y$ makes the posterior of $\bDe$ to a mixture, as can be read in the figure.

\bigskip
Next, we computed the matrix mean squared errors (MSEs) of posterior means for $\bOm_y$ and $\bDe$, based on the four MCMC algorithms. 
To see the effect of the number of MCMC samples on the MSE, we show the MSEs computed at every 5000 iterations under $p=10$ and $q=3, 7$ and $15$ in Figure~\ref{fig:PGGM-MSE}.
As expected from the (in)efficiency observed in Figure~\ref{fig:PGGM-plot}, the MH method has significantly higher MSEs than the GS method does for all the parameters, even in the cases of longer Markov chains. The HR method can improve the accuracy of estimation by running the algorithm longer, but 30000 iterations are still not enough to be competitive with the GS method. The MI method, or the ad-hoc plug-in approach, results in the worst MSEs, highlighting the importance of formally quantifying the posterior uncertainty of $\bOm_y$.

\bigskip
Finally, we check the averaged ESSs (scaled by computation time) of $\bOm_y$ and $\bDe$ for the GS and HR methods, computing the median of 100 replications and summarizing them as the function of $q$ in Figure~\ref{fig:PGGM-ESS}. 
Note that $\bOm_y$ and $\bDe$ are $q{\times}q$ and $q{\times}p$ matrices, respectively. The GS method outperforms the HR method for both parameters, and its difference in ESSs grows as $q$ increases. To sum, we confirm in this example that the use of the Gibbs sampler is strongly advised in applications that involve the MGIG distributions.

\begin{figure}[!htb]
\centering
\includegraphics[width = \linewidth]{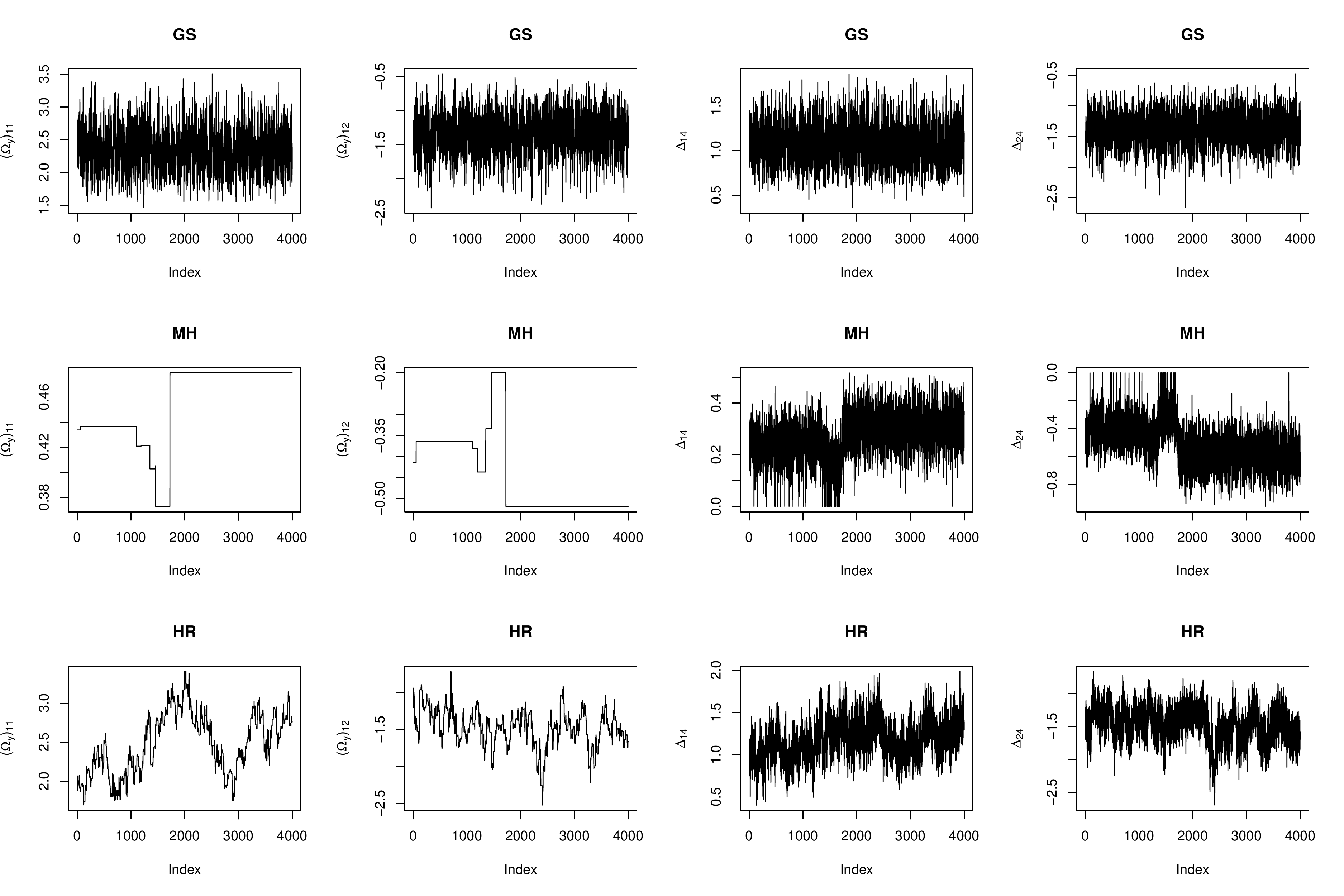}
\caption{Traceplots of $(\bOm_y)_{11}$, $(\bOm_y)_{12}$, $\bDe_{14}$ and $\bDe_{24}$ obtained by the MCMC algorithm with three different MGIG samplers, the proposed Gibbs sampler (GS), independent MH algorithm (MH) and hit-and-run sampler (HR) under $q=3$ and $p=10$. (The original 20000 scans are thinned to 4000 for this figure.)}
\label{fig:PGGM-plot}
\end{figure}

\begin{figure}[!htb]
\centering
\includegraphics[width = \linewidth]{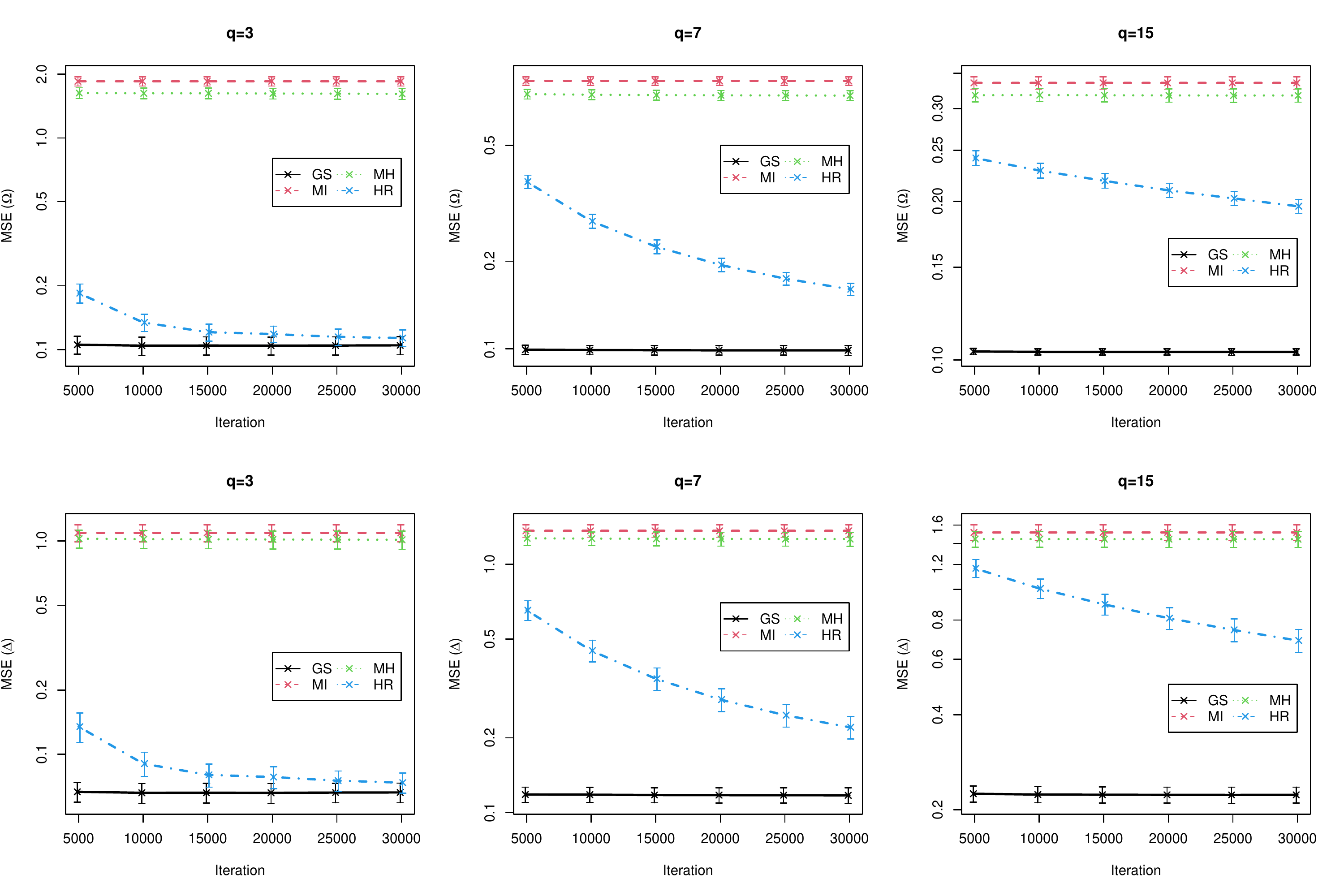}
\caption{Matrix mean squared errors (MSEs) of the posterior means of the MCMC algorithms with four different MGIG samplers, as a function of the number of MCMC iterations, under three choices of $q$ (dimension of $\bOm_y$) and $p=10$. }
\label{fig:PGGM-MSE}
\end{figure}

\begin{figure}[!htb]
\centering
\includegraphics[width = \linewidth]{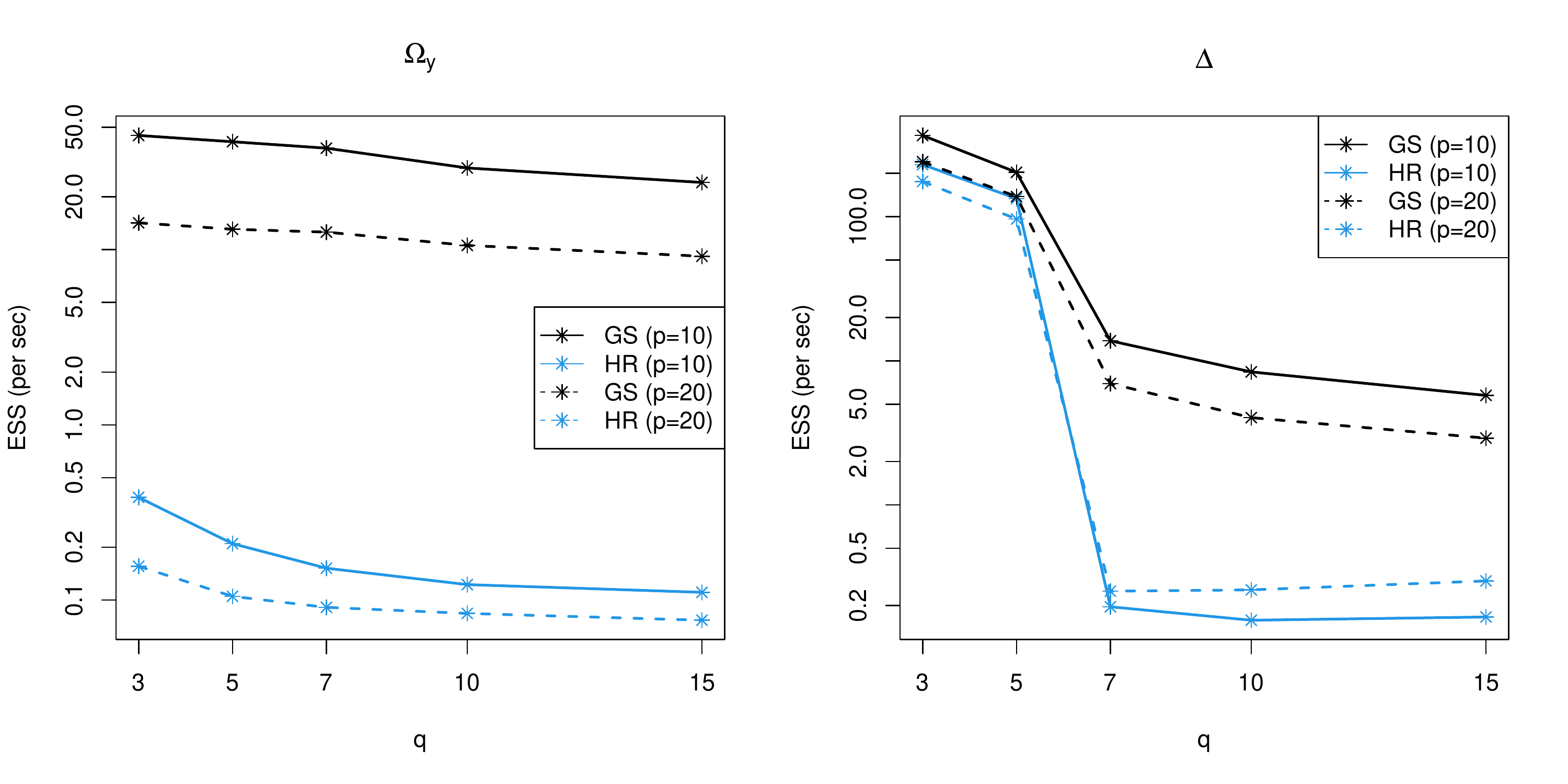}
\caption{Median of ESSs of the GS and HR methods replicated for 100 times, under five choices of $q$ (dimension of $\bOm_y$) and two choices of $p$.}
\label{fig:PGGM-ESS}
\end{figure}

\subsection{Matrix skewed-$t$ distributions}\label{sec:num-MST}
As seen in the graphical model of Section~\ref{sec:pggm}, a typical class of statistical models where the MGIG distributions naturally arise is the mean-variance mixture of multivariate/matrix-variate normal distributions. However, in the literature, such multivariate models are often limited to the mixtures by scaler latent variables for simplicity and computational feasibility. Examples include the multivariate generalized hyperbolic distributions \cite{protassov2004based} and matrix skew-$t$ distributions \citep{gallaugher2017matrix}. In this subsection, we consider the Wishart mixture of matrix-variate normals as the extension of the aforementioned matrix skew-$t$ model, the posterior inference of which is enabled by the proposed Gibbs sampler. 

For $p\times q$ matrix observations $\Y_1,\ldots,\Y_n$, we define the matrix skew-$t$ model as the following matrix mixture: 
\begin{equation}\label{MST}
\Y_i|\W_i\sim N_{p,q}(\M + \W_i \B , \W_i, \bOm), \ \ \ \ \W_i\sim {\rm IW}_p(\bPsi, \nu), \ \ \ i=1,\ldots,n,
\end{equation}
where $\M$ and $\B$ are $p\times q$ matrix parameters representing mean and skewness parameters, respectively, $\bOm$ is a $q\times q$ covariance matrix, and $\nu$ and $\bPsi$ are the scalar and $p\times p$ positive definite matrix parameters of the Wishart distribution, respectively. 
For identifiability, the $(1,1)$-entry of $\bPsi$ is set to unity. 
Here $\W_i$ is a $p\times p$ latent matrix.
Note that, when $\B=\O$, the marginal model (\ref{MST}) reduces to the matrix-$t$ distribution \citep[e.g.][]{dawid1981some, thompson2020classification}.

In what follows, we fix $\nu$ and introduce prior distributions for the other parameters: $\M\sim N_{p,q}(\A_{0M}, \U_{0M}, \V_{0M})$, $\B\sim N_{p,q}(\A_{0B}, \U_{0B}, \V_{0B})$, $\bPsi\sim {\rm W}_p(\bPsi_0, \eta_0)$ and $\bOm\sim {\rm IW}_q(\bOm_0, \xi_0)$.
Then, the full conditional distributions of the latent matrix $\W_i^{-1}$ is ${\rm MGIG}_p((\nu+q-p-1)/2, \tilde{\bGa}_i, \tilde{\bPhi}_i)$, where 
$$
\tilde{\bPhi}_i=\B\bOm^{-1}\B^\top, \ \ \ \ 
\tilde{\bGa}_i=\bPsi + (\Y_i-\M)\bOm^{-1}(\Y_i-\M)^\top.
$$
The details of the other full conditional distributions are given in the Supplementary Material (Section~S6).

To illustrate the matrix skew-$t$ (MST) model, we take the landsat satellite data analyzed in \cite{thompson2020classification}. 
This multi-spectral satellite imagery data \citep{Dua:2019} records images in two visible and two infrared bands ($q=4$) on $3\times 3$ pixel segments ($q=9$), yielding $4\times 9$ matrix observations. These observations are labeled according to the terrain types, resulting three datasets: cotton crop ($n=479$), gray soil ($n=961$) and soil with vegetation stubble segments ($n=470$). The MST model (\ref{MST}) is fitted to each of the three models individually. 
For comparison, we also fitted a matrix $t$ (MT) distribution \citep[e.g.][]{dawid1981some, thompson2020classification} to see the benefit of the skewness introduced in (\ref{MST}). 

For each class of the satellite imagery data, we fit both MST and MT models with $\nu =5$ and $10$.
In applying the MST model, we use three samplers (GS, MH and HR) to generate the latent matrix $\W_i$.
Note that the MCMC algorithms for fitting the MT model does not require sampling from the MGIG distribution. 
In each algorithm, we generated 5,000 posterior samples after discarding the first 1,000 samples. 
First, we compute posterior predictive loss \citep{gelfand1998model} of the MST and MT models based on the outputs of Gibbs samplers, and report the results Table~\ref{tab:Sat-PPL}. 
It shows that the MST model with $\nu=10$ attains the smallest posterior predictive loss, indicating the improved model fit to this dataset by introducing the skewness structures. 
In Table~\ref{tab:Sat-ESS}, we present ESSs of the MCMC algorithms with three different samplers for the MGIG distribution. 
Unlike the results in the previous section, the MH method performs reasonably well compared with the HR method, and is even competitive with the GS method in a few cases. Still, the ESSs of the GS methods are significantly better than those of the other two methods in most of the data analyses. 
\begin{table}[htbp!]
\caption{Posterior predictive loss of MST and MT models with two choices of degrees of freedom, $\nu=5$ and $10$. 
}
\label{tab:Sat-PPL}
\begin{center}
\begin{tabular}{ccccccccccc}
\hline
Class &  & MST($\nu=5$) & MST($\nu=10$) & MT($\nu=5$) & MT($\nu=10$) \\
\hline
cotton crop &  & 341 & 288 & 523 & 452 \\
gray soil &  & 101 & 90 & 112 & 100 \\
vegetation &  & 266 & 231 & 307 & 272 \\
\hline
\end{tabular}
\end{center}
\end{table}

\begin{table}[htbp!]
\caption{The effective sample size (ESS) of the matrix parameters in MST model under Gibbs sampler (GS), Metropolis-Hastings  algorithm (MH) with a Wishart proposal and Hit-and-Run sampler(HR) for the MGIG full conditional distribution. 
}
\label{tab:Sat-ESS}
\begin{center}
\begin{tabular}{ccccccccccc}
\hline
&&& \multicolumn{3}{c}{$\nu=5$} &&  \multicolumn{3}{c}{$\nu=10$}\\
Parameter & Class &  & GS & MH & HR &  & GS & MH & HR \\
 \hline
 & cotton crop &  & 843 & 419 & 39 &  & 877 & 283 & 33 \\
$\W$ & gray soil  &  & 1523 & 1200 & 44 &  & 1522 & 1135 & 39 \\
 & vegetation &  & 1821 & 1425 & 49 &  & 1527 & 1001 & 38 \\
 \hline
 & cotton crop &  & 637 & 551 & 103 &  & 369 & 301 & 85 \\
$\B$ & gray soil  &  & 699 & 692 & 124 &  & 444 & 429 & 98 \\
 & vegetation &  & 708 & 671 & 133 &  & 447 & 443 & 95 \\
 \hline
 & cotton crop &  & 1655 & 1678 & 165 &  & 1139 & 755 & 62 \\
$\bPsi$ & gray soil  &  & 1889 & 1724 & 160 &  & 1278 & 1091 & 61 \\
 & vegetation &  & 1796 & 1799 & 142 &  & 1273 & 1037 & 64 \\
 \hline
 & cotton crop &  & 1151 & 935 & 114 &  & 791 & 452 & 71 \\
$\bOm$ & gray soil  &  & 1812 & 1731 & 182 &  & 2047 & 1727 & 172 \\
 & vegetation &  & 1268 & 1131 & 61 &  & 1162 & 893 & 68 \\
\hline
\end{tabular}
\end{center}
\end{table}

\section{Concluding Remarks}
\label{sec:conclusion}

Sampling from the MGIG distribution is often an unavoidable step of posterior computation in many statistical models, but is rarely discussed as the main computational issue. Some ad-hoc alternatives to the exact sampling, such as plugging the point estimate, have been practiced, but could disprove both the sampling efficiency and the accuracy of posterior computation significantly deteriorated, as we observed in Section~\ref{sec:pggm}. Our Gibbs sampler is an answer to this computational problem, enabling the precise implementation of the MCMC methods for the models involving the MGIG distributions.

\appendix 

\section{Appendix}

\subsection{Sampling from MGIG distributions with degenerate matrix parameters}
\label{app:deg}

\cite{MASSAM2006103} showed the following property of the MGIG distribution with degenerate matrix parameters, known as the Matsumoto-Yor property.

\begin{thm}[\citealt{MASSAM2006103}]
\label{lem:MW_2} 
Let $p, q \in \mathbb{N}$. 
Let $\la > - 1$ and let $\bTh $ be a $p \times q$ matrix of full rank. 
Let $\bPsi $ be a $p \times p$ positive definite matrix. 
Suppose that 
\begin{align}
\X \sim {\rm{MGIG}}_q (- \la - 1 - q, \bTh ^{\top } \bPsi \bTh , \I ) \quad \text{and} \quad \Y \sim {\rm{W}}_p ( 2\la + p + 1, \bPsi ^{- 1} ) \non 
\end{align}
are independent. 
Then 
\begin{align}
\bTh \X \bTh ^{\top } + \Y \sim {\rm{MGIG}}_p ( \la , \bPsi , \bTh \bTh ^{\top } ) \text{.} \non 
\end{align}
\end{thm}

When $\bSi \sim {\rm{MGIG}}_p(\lambda ,\bPsi ,\bGa)$ and the rank of $\bGa$ is $q$ ($q<p$), one can consider the decomposition of $\bGa = \bTh \bTh^{\top}$ for some full-rank $p\times q$ matrix $\bTh$, and sample $\bSi$ by simulating $\X$ and $\Y$ as described above and setting $\bSi = \bTh\X\bTh^{\top}+\Y$. Then, the problem reduces to the simulation from ${\rm{MGIG}}_q(-\lambda - 1 -q, \bTh^{\top} \bPsi \bTh, \I )$, the MGIG distribution with full-rank matrix parameters, which is covered in this article. 
The case of degenerate $\bPsi $ can be discussed similarly. %

\cite{fang2020bayesian} utilize the Matsumoto-Yor property of the MGIG distributions and enable the direct sampling from the MGIG distribution when $q = 1$. Combined with this idea, the MH and Gibbs sampler proposed in this article can be extended to an arbitrary MGIG distribution.

\section*{Acknowledgments}
Research of the authors was supported in part by JSPS KAKENHI Grant Number 22K20132, 19K11852, 17K17659, and 21H00699 from Japan Society for the Promotion of Science.

\vspace{1cm}
\bibliographystyle{chicago}
\bibliography{mgig}

\newpage
\setcounter{page}{1}
\setcounter{equation}{0}
\renewcommand{\theequation}{S\arabic{equation}}
\setcounter{section}{0}
\renewcommand{\thesection}{S\arabic{section}}
\setcounter{lem}{0}
\renewcommand{\thelem}{S\arabic{lem}}
\setcounter{thm}{0}
\renewcommand{\thethm}{S\arabic{thm}}
\setcounter{table}{0}
\renewcommand{\thetable}{S\arabic{table}}
\setcounter{figure}{0}
\renewcommand{\thetable}{S\arabic{figure}}

\renewcommand{\thefigure}{S\arabic{figure}}
\renewcommand{\thetable}{S\arabic{table}}

\begin{center}
{\LARGE {\bf Supplementary Materials for ``Gibbs Sampler for Matrix Generalized Inverse Gaussian Distributions''}}
\end{center}

\vspace{1cm}
This Supplementary Materials provide theoretical details of the main document and additional simulation results.
In Section~S1, we state Theorem~1 in Section~3 precisely. Then, in Section~S2, we prove the theorem to provide the full conditional distributions used in Algorithm~1. 
In Section~S3, we prove the results on the limit of the average acceptance rates of the MH method. 
In Section~S4, we explain the possible improvement of the Gibbs sampler by parallelization. 
In Section~S5, we report the additional results about the simulation study in Section~4.1. 
In Section~S6, we summarize the Gibbs sampler for the matrix-skew-$t$ distributions used in Section~4.3. 

\ 

\noindent
\underline{\textbf{Notations:}}

\begin{itemize}
\item
For any $m \in \mathbb{N}$, we write $\O ^{(m)}$ and $\I ^{(m)}$ for the $m \times m$ zero and identity matrices, respectively. 
\item
For any $m, n \in \mathbb{N}$, we write $\O ^{(m, n)}$ for the $m \times n$ zero matrix. 
\item
For any $m \in \mathbb{N}$, we write $\bm{0} ^{(m)}$ for the $m$-dimensional zero vector. 
\item
For any $m \in \mathbb{N}$, we write $\e _{i}^{(m)}$ for the $i$-th column vector of $\I ^{(m)}$ for $i = 1, \dots , m$. 
\item
For any $m \in \mathbb{N}$, we write $\E _{i, j}^{(m)} = \e _{i}^{(m)} ( \e _{j}^{(m)} )^{\top }$ for $i, j = 1, \dots , m$. 
\item
As in the main text, for any $m \in \mathbb{N}$, if $\c _1 , \dots , \c _m$ are vectors, we write $( \c _i )_{i = 1}^{m}$ for $( {\c _1}^{\top } , \dots , {\c _m}^{\top } )^{\top }$. 
\item
As in the main text, for any $m, n \in \mathbb{N}$, if $\C $ is an $m \times n$ matrix and if $c_{i, j}$ is the $(i, j)$-th element of $\C $ for $i = 1, \dots , m$ and $j = 1, \dots , n$, we write $( \C )_{\underline{i}:{\overline{i}}, \underline{j}:{\overline{j}}}$ for the submatrix 
\begin{align}
\begin{pmatrix} c_{\underline{i} , \underline{j}} & \cdots & c_{\underline{i}, \overline{j}} \\ \vdots & \ddots & \vdots \\ c_{\overline{i} , \underline{j}} & \cdots & c_{\overline{i} , \overline{j}} \end{pmatrix} \non 
\end{align}
for $1 \le \underline{i} \le \overline{i} \le m$ and $1 \le \underline{j} \le \overline{j} \le n$. 
\item
As in the main text, for any $m \in \mathbb{N}$, if $a_1 , \dots , a_m > 0$ and if $\A = \bdiag ( a_1 , \dots , a_m )$, we write $\A ^{1 / 2} = \bdiag ( \sqrt{a_1} , \dots , \sqrt{a_m} )$ and $\A ^{- 1 / 2} = \bdiag ( 1 / \sqrt{a_1} , \dots , 1 / \sqrt{a_m} )$. 
\end{itemize}

\section{Full conditional distributions}

\begin{thm}
\label{thm:conditional_density} 
\hfill
\begin{itemize}
\item[{\rm{(i)}}]
The joint density of $\a $ and $\b $ is 
\begin{align}
p( \a , \b ) &\propto %
\Big( \prod_{i = 1}^{p} {a_i}^{\la + p - i} \Big) \exp [- \tr \{ \A ^{1 / 2} \B ^{\top } \bPsi \B \A ^{1 / 2} + \A ^{- 1 / 2} \B ^{- 1} \bGa ( \B ^{- 1} )^{\top } \A ^{- 1 / 2} \} / 2] \text{.} \non 
\end{align}
\item[{\rm{(ii)}}]
The conditional distribution of $\a $ given $\b $ is 
\begin{align}
p( \a | \b ) &= \prod_{i = 1}^{p} {\rm{GIG}} ( a_i | \la + p - i + 1, ( \B ^{\top } \bPsi \B )_{i, i} , ( \B ^{- 1} \bGa ( \B ^{- 1} )^{\top } )_{i, i} ) \text{.} \non 
\end{align}
\item[{\rm{(iii)}}]
Let 
\begin{align}
&\B _i = \begin{pmatrix} \e _{1}^{(p)} & \cdots & \e _{i - 1}^{(p)} & \begin{pmatrix} \bm{0} ^{(i - 1)} \\ 1 \\ \b _i \end{pmatrix} & \e _{i + 1}^{(p)} & \cdots & \e _{p}^{(p)} \end{pmatrix} \quad \text{and} \non \\
&\overline{\B } _i = \begin{pmatrix} \e _{1}^{(p)} & \cdots & \e _{i - 1}^{(p)} & \begin{pmatrix} \bm{0} ^{(i - 1)} \\ 1 \\ - \b _i \end{pmatrix} & \e _{i + 1}^{(p)} & \cdots & \e _{p}^{(p)} \end{pmatrix} = 2 \I ^{(p)} - \B _i \non 
\end{align}
for $i = 1, \dots , p$. 
Let 
\begin{align}
&\M _i = {\B _{i - 1}}^{\top } \dotsm {\B _1}^{\top } \bPsi \B _1 \dotsm \B _{i - 1} \text{,} \non \\
&\R _i = \B _{i + 1} \dotsm \B _p \A ^{1 / 2} \text{,} \non \\
&\overline{\M } _{i} = \overline{\B }_{i - 1} \dotsm \overline{\B }_{1} \bGa {\overline{\B }_{1}}^{\top } \dotsm {\overline{\B }_{i - 1}}^{\top } \text{,} \quad \text{and} \non \\
&\overline{\R }_{i} = {\overline{\B }_{i + 1}}^{\top } \dotsm {\overline{\B }_{p}}^{\top } \A ^{- 1 / 2} \non 
\end{align}
for $i = 1, \dots , p$. 
Then for each $i = 1, \dots , p - 1$, %
the conditional distribution of $\b _i$ given $\a $ and $\b _{- i} = \b \setminus \b _i$ is multivariate normal with 
variance 
\begin{align}
\N _i &= a_i ( \bPsi ) _{(i + 1):p, (i + 1):p} + ( \overline{\M }_{i} )_{i, i} (( \B ^{- 1} )^{\top } \A ^{- 1} \B ^{- 1} ) _{(i + 1):p, (i + 1):p} \text{,} \non 
\end{align}
which is independent of $\b _i$, and mean 
\begin{align}
{\N _i}^{- 1} \{ - ( \M _i )_{(i + 1):p, 1:p} \R _i (( \R _i )_{i, 1:p} )^{\top } + ( \overline{\R }_{i} )_{(i + 1):p, 1:p} {\overline{\R }_{i}}^{\top } ( {( \overline{\M }_{i} )_{i, 1:p}} )^{\top } \} \text{.} \non 
\end{align}
\end{itemize}
\end{thm}

\section{Proof of Theorem \ref{thm:conditional_density}}

\begin{proof}%
For part (i), 
\begin{align}
p( \a , \b ) &\propto %
\Big( \prod_{i = 1}^{p} {a_i}^{\la + p - i} \Big) \exp [- \tr \{ \bPsi \B \A \B ^{\top } + \bGa ( \B \A \B ^{\top } )^{- 1} \} / 2] \non \\
&= \Big( \prod_{i = 1}^{p} {a_i}^{\la + p - i} \Big) \exp [- \tr \{ \A ^{1 / 2} \B ^{\top } \bPsi \B \A ^{1 / 2} + \A ^{- 1 / 2} \B ^{- 1} \bGa ( \B ^{- 1} )^{\top } \A ^{- 1 / 2} \} / 2] \text{.} \non 
\end{align}
For part (ii), 
\begin{align}
p( \a | \b ) &\propto \prod_{i = 1}^{p} ( {a_i}^{\la + p - i} \exp [- \{ ( \B ^{\top } \bPsi \B )_{i, i} a_i + ( \B ^{- 1} \bGa ( \B ^{- 1} )^{\top } )_{i, i} / a_i \} / 2] ) \non \\
&\propto \prod_{i = 1}^{p} {\rm{GIG}} ( a_i | \la + p - i + 1, ( \B ^{\top } \bPsi \B )_{i, i} , ( \B ^{- 1} \bGa ( \B ^{- 1} )^{\top } )_{i, i} ) \text{.} \non 
\end{align}
For part (iii), note that 
\begin{align}
\B = \B _1 \dotsm \B _p \text{.} \non
\end{align}
Since $\B_i^{-1} = \overline{\B}_i$ for all $i$, we also have 
\begin{align}
\B ^{- 1} = \overline{\B }_p \dotsm \overline{\B } _1 \text{.} \non 
\end{align}
Fix $i = 1, \dots , p - 1$. 
Then 
\begin{align}
p( \b _i | \a , \b _{- i} ) %
&\propto \exp \{ - \tr ( {\R _i}^{\top } {\B _i}^{\top } \M _i \B _i \R _i ) / 2 - \tr ( {\overline{\R }_i}^{\top } \overline{\B } _i \overline{\M } _i {\overline{\B } _i}^{\top } \overline{\R } _i ) / 2 \} \text{.} \non 
\end{align}
Now we write the column vectors of $\R_i$ and $\overline{\R}_i$ as $( \r _{i, 1} , \dots , \r _{i, p} ) = \R _i$ and $( \overline{\r } _{i, 1} , \dots , \overline{\r } _{i, p} ) = \overline{\R } _i$. Note that the $k$-th element of vector $\r_{i,j}$ is also written as $(\R_i)_{k,j}$. 
Then 
\begin{align}
p( \b _i | \a , \b _{- i} ) &\propto \exp \Big\{ - {1 \over 2} \Big( \sum_{j = 1}^{p} {\r _{i, j}}^{\top } {\B _i}^{\top } \M _i \B _i \r _{i, j} + \sum_{j = 1}^{p} {\overline{\r } _{i, j}}^{\top } \overline{\B } _i \overline{\M } _i {\overline{\B } _i}^{\top } \overline{\r } _{i, j} \Big) \Big\} \text{.} \non 
\end{align}
To write the density above as the function of $\b_i$, observe that 
\begin{align}
\B _i \r _{i, j} = \r _{i, j} + (\R _i )_{i, j} \begin{pmatrix} \bm{0} ^{(i)} \\ \b_i \end{pmatrix} = \begin{pmatrix} \r _{i, j} & \begin{pmatrix} \O ^{(i, p - i)} \\ ( \R _i )_{i, j} \I ^{(p - i)} \end{pmatrix} \end{pmatrix} \begin{pmatrix} 1 \\ \b _i \end{pmatrix} \non 
\end{align}
and 
\begin{align}
{\overline{\B } _i}^{\top } \overline{\r } _{i, j} = \begin{pmatrix} \overline{\r } _{i, j} & \begin{pmatrix} \O ^{(i - 1, p - i)} \\ - \{ ( \overline{\R } _i )_{(i + 1):p, j} \} ^{\top } \\ \O ^{(p - i, p - i)} \end{pmatrix} \end{pmatrix} \begin{pmatrix} 1 \\ \b _i \end{pmatrix} \text{.} \non 
\end{align}
Then, for all $j = 1, \dots , p$, we have 
\begin{align}
p( \b _i | \a , \b _{- i} ) &\propto \exp \Big( - {1 \over 2} \begin{pmatrix} 1 & {\b _i}^{\top } \end{pmatrix} \begin{pmatrix} n_i & - {\n _i}^{\top } \\ - \n _i & \N _i \end{pmatrix} \begin{pmatrix} 1 \\ \b _i \end{pmatrix} \Big) \non \\
&\propto {\rm{N}}_{p - i} ( \b _i | {\N _i}^{- 1} \n _i , {\N _i}^{- 1} ) \text{,} \non 
\end{align}
where 
\begin{align}
\begin{pmatrix} n_i & - {\n _i}^{\top } \\ - \n _i & \N _i \end{pmatrix} &= \sum_{j = 1}^{p}  \begin{pmatrix} {\r _{i, j}}^{\top } \\ \begin{pmatrix} \O ^{(p - i, i)} & ( \R _i )_{i, j} \I ^{(p - i)} \end{pmatrix} \end{pmatrix} \M _i \begin{pmatrix} \r _{i, j} & \begin{pmatrix} \O ^{(i, p - i)} \\ ( \R _i )_{i, j} \I ^{(p - i)} \end{pmatrix} \end{pmatrix} \non \\
&\quad + \begin{pmatrix} {\overline{\r } _{i, j}}^{\top } \\ \begin{pmatrix} \O ^{(p - i, i - 1)} & - ( \overline{\R } _i )_{(i + 1):p, j} & \O ^{(p - i, p - i)} \end{pmatrix} \end{pmatrix} \overline{\M } _i \begin{pmatrix} \overline{\r } _{i, j} & \begin{pmatrix} \O ^{(i - 1, p - i)} \\ - \{ ( \overline{\R } _i )_{(i + 1):p, j} \} ^{\top } \\ \O ^{(p - i, p - i)} \end{pmatrix} \end{pmatrix} \text{.} \non 
\end{align}
Here, 
\begin{align}
- \n _i &= \sum_{j = 1}^{p} \Big( \begin{pmatrix} \O ^{(p - i, i)} & ( \R _i )_{i, j} \I ^{(p - i)} \end{pmatrix} \M _i \r _{i, j} + \begin{pmatrix} \O ^{(p - i, i - 1)} & - ( \overline{\R } _i )_{(i + 1):p, j} & \O ^{(p - i, p - i)} \end{pmatrix} \overline{\M } _i \overline{\r } _{i, j} \Big) \non \\
&= \sum_{j = 1}^{p} \Big( ( \R _i )_{i, j} ( \M _i )_{(i + 1):p, 1:p} \r _{i, j} - ( \overline{\R } _i )_{(i + 1):p, j} ( \overline{\M } _i )_{i, 1:p} \overline{\r } _{i, j} \Big) \non \\
&= \sum_{j = 1}^{p} \Big\{ ( \M _i )_{(i + 1):p, 1:p} \r _{i, j} ( \R _i )_{i, j} - ( \overline{\R } _i )_{(i + 1):p, j} {\overline{\r } _{i, j}}^{\top } ( {( \overline{\M } _i )_{i, 1:p}} )^{\top } \Big\} \non \\
&= ( \M _i )_{(i + 1):p, 1:p} \R _i (( \R _i )_{i, 1:p} )^{\top } - ( \overline{\R } _i )_{(i + 1):p, 1:p} {\overline{\R } _i}^{\top } ( {( \overline{\M } _i )_{i, 1:p}} )^{\top } \text{.} \non 
\end{align}
Meanwhile, 
\begin{align}
\N _i &= \sum_{j = 1}^{p} \Big[ \begin{pmatrix} \O ^{(p - i, i)} & ( \R _i )_{i, j} \I ^{(p - i)} \end{pmatrix} \M _i \begin{pmatrix} \O ^{(i, p - i)} \\ ( \R _i )_{i, j} \I ^{(p - i)} \end{pmatrix} \non \\
&\quad + \begin{pmatrix} \O ^{(p - i, i - 1)} & - ( \overline{\R } _i )_{(i + 1):p, j} & \O ^{(p - i, p - i)} \end{pmatrix} \overline{\M } _i \begin{pmatrix} \O ^{(i - 1, p - i)} \\ - \{ ( \overline{\R } _i )_{(i + 1):p, j} \} ^{\top } \\ \O ^{(p - i, p - i)} \end{pmatrix} \Big] \non \\
&= \sum_{j = 1}^{p} \Big[ (( \R _i )_{i, j} )^2 ( \M _i )_{(i + 1):p, (i + 1):p} + ( \overline{\M } _i )_{i, i} ( \overline{\R } _i )_{(i + 1):p, j} \{ ( \overline{\R } _i )_{(i + 1):p, j} \} ^{\top } \Big] \non \\
&= \| ( \R _i )_{i, 1:p} \| ^2 ( \M _i )_{(i + 1):p, (i + 1):p} + ( \overline{\M } _i )_{i, i} ( \overline{\R } _i )_{(i + 1):p, 1:p} \{ ( \overline{\R } _i )_{(i + 1):p, 1:p} \} ^{\top } \text{.} \non 
\end{align}
We have 
\begin{align}
\| ( \R _i )_{i, 1:p} \| ^2 &= %
\Big\| ( \e _{i}^{(p)} )^{\top } \begin{pmatrix} \I ^{(i)} & \O ^{(i, p - i)} \\ \O ^{(p - i, i)} & ( \B )_{(i + 1):p, (i + 1):p} \end{pmatrix} \A ^{1 / 2} \Big\| ^2 = a_i \non 
\end{align}
and 
\begin{align}
&( \M _i )_{(i + 1):p, (i + 1):p} \non \\
&= \begin{pmatrix} \O ^{(p - i, i)} & \I ^{(p - i)} \end{pmatrix} \begin{pmatrix} (( \B )_{1:p, 1:(i - 1)} )^{\top } \\ \begin{pmatrix} \O ^{(p - i + 1, i - 1)} & \I ^{(p - i + 1)} \end{pmatrix} \end{pmatrix} \bPsi \begin{pmatrix} ( \B )_{1:p, 1:(i - 1)} & \begin{pmatrix} \O ^{(i - 1, p - i + 1)} \\ \I ^{(p - i + 1)} \end{pmatrix} \end{pmatrix} \begin{pmatrix} \O ^{(i, p - i)} \\ \I ^{(p - i)} \end{pmatrix} \non \\
&= ( \bPsi )_{(i + 1):p, (i + 1):p} \text{.} \non 
\end{align}
Furthermore, %
\begin{align}
&( \overline{\R } _i )_{(i + 1):p, 1:p} \{ ( \overline{\R } _i )_{(i + 1):p, 1:p} \} ^{\top } \non \\
&= \begin{pmatrix} \O ^{(p - i, i)} & \I ^{(p - i)} \end{pmatrix} {\overline{\B } _{i + 1}}^{\top } \dotsm {\overline{\B } _p}^{\top } \A ^{- 1} \overline{\B } _p \dotsm \overline{\B } _{i + 1} \begin{pmatrix} \O ^{(i, p - i)} \\ \I ^{(p - i)} \end{pmatrix} \non \\
&= \begin{pmatrix} \O ^{(p - i, i)} & \I ^{(p - i)} \end{pmatrix} ( \B ^{- 1} )^{\top } \A ^{- 1} \B ^{- 1} \begin{pmatrix} \O ^{(i, p - i)} \\ \I ^{(p - i)} \end{pmatrix} \non \\
&= (( \B ^{- 1} )^{\top } \A ^{- 1} \B ^{- 1} )_{(i + 1):p, (i + 1):p} \text{.} \non 
\end{align}
Therefore, 
\begin{align}
\N _i &= a_i ( \bPsi ) _{(i + 1):p, (i + 1):p} + ( \overline{\M } _i )_{i, i} (( \B ^{- 1} )^{\top } \A ^{- 1} \B ^{- 1} )_{(i + 1):p, (i + 1):p} \text{.} \non 
\end{align}
This completes the proof. 
\end{proof}

\section{Average acceptance rate of the first MH method}

In this section, we compute the limit of the average acceptance rate when using the first MH method (MH1) in the two examples in the main text. The average acceptance rate is defined as 
\begin{align}
{\rm{AAR}}(\la,\bPsi,\bGa) = 2 \mathbb{P}[ \tr ( \bGa {\bSi _{\rm{new}}}^{- 1} ) \le \tr ( \bGa {\bSi _{\rm{old}}}^{- 1} ) ] \text{,} \non 
\end{align}
where $\bSi _{\rm{new}}$ and $\bSi _{\rm{old}}$ are independent and have densities
\begin{align}
&p(\bSi _{\rm{new}}) \propto | \bSi |^{\la } \etr (- \bPsi \bSi / 2) \text{,} \non \\
&p(\bSi _{\rm{old}}) \propto | \bSi |^{\la } \etr \{ - ( \bPsi \bSi + \bGa \bSi ^{- 1} ) / 2 \} \text{,} \non 
\end{align}
and we assume $\la > -1$ so that $p(\bSi _{\rm{old}})$ becomes a proper probability density.

\subsection{Example~1: Small and large $\la$}

Although we stated in the main text that $\bPsi$ and $\bGa$ are arbitrary, in the proofs below, we set either $\bPsi$ or $\bGa$ to $2\I^{(p)}$ without loss of generality; see the discussions in Section~2.1 in the main text. 

\begin{prp}
\label{prp:aar_la_0} 
Suppose without loss of generality that $\bPsi / 2 = \bdiag ( \psi _1 , \dots , \psi _p )$ and $\bGa / 2 = \I ^{(p)}$. 
Then, as $\la \to - 1$, the average acceptance rate converges to $0$. 
\end{prp}

\begin{proof}
Note that 
\begin{align}
( {\bSi _{\rm{new}}}^{- 1} )_{i, i} \sim {\rm{IG}} ( \la + 1, \psi _{i, i} ) \non 
\end{align}
for all $i = 1, \dots , p$ (see, for example, \cite{wwc2018}). 
Since $( {\bSi _{\rm{new}}}^{- 1} )_{1, 1} \le \tr ( \bSi _{\rm{new}} ^{- 1} )$, we have
\begin{align}
{\rm{AAR}}(\la,\bPsi,\bGa) &\le 2 \mathbb{P}[ ( {\bSi _{\rm{new}}}^{- 1} )_{1, 1} \le \tr ( {\bSi _{\rm{old}}}^{- 1} ) ] \non \\
&\le %
2 \mathbb{E}[ \tr ( {\bSi _{\rm{old}}}^{- 1} ) ] ( \la + 1) / \psi _{1, 1} \text{,} \non 
\end{align}
where the second inequality follows from the conditional Markov's inequality. Here the expectation, $\mathbb{E}[ \tr ( {\bSi _{\rm{old}}}^{- 1} ) ]$, depends on $\la$, but converges to $\mathbb{E}[ \tr ( \bSih^{- 1} ) ]$ as $\la \downarrow -1$, where $\bSih \sim {\rm{MGIG}}_p (- 1, \bPsi , \bGa )$. Note that $\mathbb{E}[ \bSih^{- 1} ]$ is shown to be finite. Thus, ${\rm{AAR}}(\la,\bPsi,\bGa) \to 0$ as $\la \downarrow - 1$. 

To see that $\mathbb{E}[ \tr ( \bSih^{- 1} ) ]$ is finite, use the dominated convergence theorem as follows. Let $\la \in \mathbb{R}$ and let $\bPsi , \bGa > \O ^{(p)}$. The density of $\bSi$ is proportional to $| \bSi |^{\la } \etr (- \bPsi \bSi - \bGa \bSi ^{- 1} )$. 
Choose $c > 0$ such that $c \I ^{(p)} < \bPsi , \bGa $. 
Then, using the nomralizing constant $c_p(\la,\bPsi,\bGa)$, we have 
\begin{align}
c_p(\la , \bPsi , \bGa )  E[ \tr ( \bSi ) ] &\le \int ( \tr \bSi ) | \bSi |^{\la } \etr (- c \bSi - c \bSi ^{- 1} ) d\bSi \non \\
&\le \int ( \tr \bSi ) (| \bSi |^{| \la |} + 1 / | \bSi |^{| \la |} ) \etr (- c \bSi - c \bSi ^{- 1} ) d\bSi \non \\
&\le \Big[ \sup_{\de _1 , \dots , \de _p > 0} \Big\{ \Big( \prod_{i = 1}^{p} {\de _i}^{| \la |} + \prod_{i = 1}^{p} {1 \over {\de _i}^{| \la |}} \Big) \prod_{i = 1}^{p} \exp \Big( - {c \over 2} \de _i - {c \over 2} {1 \over \de _i} \Big) \Big\} \Big] \non \\
&\quad \times \int ( \tr \bSi ) \etr \Big( - {c \over 2} \bSi - {c \over 2} \bSi ^{- 1} \Big) d\bSi < \infty \text{.} \non 
\end{align}
Thus, the trace of any MGIG-distributed matrix has a finite mean. 
\end{proof}

\begin{prp}
\label{prp:aar_la_1} 
Suppose without loss of generality that $\bPsi / 2 = \I ^{(p)}$ and $\bGa / 2 = \bdiag ( \ga _1 , \dots , \ga _p )$. 
Then, as $\la \to \infty $, the average acceptance rate converges to $1$. 
\end{prp}

\begin{proof}
In the following proof, we utilize the singular value decomposition of the positive definite random matrix. 
The change-of-variable for the MGIG distributed matrix is provided in Lemma 2 in \cite{yb1994}, which we review here. 
There exist functions $J \colon (- \pi / 2, \pi / 2)^{p (p - 1) / 2} \to (0, \infty )$ and $\bOm \colon (- \pi / 2, \pi / 2)^{p (p - 1) / 2} \to \mathbb{R} ^{p \times p}$ satisfying $J( \bom ) \le 1$, $\{ \bOm ( \bom ) \} ^{\top } \bOm ( \bom ) = \I ^{(p)}$ for all $\bom \in (- \pi / 2, \pi / 2)^{p (p - 1) / 2}$, and the following condition; if $\bde = ( \de _i )_{i = 1}^{p} \in (0, \infty )^p$ and $\bom$ are random variables and have the joint probability density,
\begin{align}
p( \bde , \bom ; \la ; \bGa ) &\propto J( \bom ) \Big\{ \prod_{1 \le i < j \le p} ( \de _i - \de _j ) \Big\} \non \\
&\quad \times \Big[ \prod_{i = 1}^{p} \{ {\de _i}^{\la } \exp (- \de _i ) \} \Big] \etr [- \bGa \{ \bOm ( \bom ) \} ^{\top } \{ \bDe ( \bde ) \} ^{- 1} \bOm ( \bom )] \mathbbm{1}( \de _1 > \dots > \de _p ) \text{,} \non 
\end{align}
where $\bDe ( \bde ) = \bdiag ( \de _1 , \dots , \de _p )$, %
then $\bSi = \{ \bOm ( \bom ) \} ^{\top } \bDe ( \bde ) \bOm ( \bom )$ follows the MGIG distribution with density proportional to $| \bSi |^{\la } \etr (- \bSi - \bGa \bSi ^{- 1} )$.
In using the lemma above, we set $\de _i = \la (1 + \xi _i / \sqrt{\la } )$. For $\xi _i \in (-\sqrt{\la},\infty)$, this is clearly one-to-one. By the change of variables, we have 
\begin{align}
p( \bxi , \bom ; \la ; \bGa ) &\propto g( \bxi , \bom ; \la ; \bGa ) \non \\
&= J( \bom ) \Big\{ \prod_{1 \le i < j \le p} ( \xi _i - \xi _j ) \Big\} \non \\
&\quad \times \etr \Big[ - {1 \over \la } \bGa \{ \bOm ( \bom ) \} ^{\top } \Big( \bdiag \Big( {1 \over 1 + \xi _1 / \sqrt{\la }}, \dots , {1 \over 1 + \xi _p / \sqrt{\la }} \Big) \Big) \bOm ( \bom ) \Big] \non \\
&\quad \times \Big( \prod_{i = 1}^{p} \exp \Big[ - \la \Big\{ {\xi _i \over \sqrt{\la }} - \log \Big( 1 + {\xi _i \over \sqrt{\la }} \Big) \Big\} \Big] \Big) \mathbbm{1}( \xi _1 > \dots > \xi _p > - \sqrt{\la } ) \text{,} \non 
\end{align}
where $\bxi = ( \xi _i )_{i = 1}^{p} \in \mathbb{R} ^p$. 
Then, we rewrite the AAR as the integral below: 
\begin{align}
{\rm{AAR}} &= 2 \mathbb{P}[ \tr ( \bGa {\bSi _{\rm{new}}}^{- 1} ) \le \tr ( \bGa {\bSi _{\rm{old}}}^{- 1} ) ] \non \\
&= 2 \mathbb{P}[  \sqrt{\la } \{ \tr \bGa - \la \tr ( \bGa {\bSi _{\rm{new}}}^{- 1} ) \} \ge \sqrt{\la } \{ \tr \bGa - \la \tr ( \bGa {\bSi _{\rm{old}}}^{- 1} ) \} ] \non \\
&= 2 \int_{\Theta^2} d( \bde _{\rm{new}} , \bom _{\rm{new}} , \bde _{\rm{old}} , \bom _{\rm{old}} ) \ g( \bde _{\rm{new}} , \bom _{\rm{new}} ; \la ; \O ^{(p)} ) \ g( \bde _{\rm{old}} , \bom _{\rm{old}} ; \la ; \bGa ) \non \\
&\quad \times \mathbbm{1}[ \ \sqrt{\la } \{ \tr \bGa - \la \tr ( \bGa [ \{ \bOm ( \bom _{\rm{new}} ) \} ^{\top } \bDe ( \bde _{\rm{new}} ) \bOm ( \bom _{\rm{new}} )]^{- 1} ) \} \non \\
&\quad \ \ge \sqrt{\la } \{ \tr \bGa - \la \tr ( \bGa [ \{ \bOm ( \bom _{\rm{old}} ) \} ^{\top } \bDe ( \bde _{\rm{old}} ) \bOm ( \bom _{\rm{old}} )]^{- 1} ) \} \ ] \non \\
&\quad / \int_{\Theta^2} d( \bde _{\rm{new}} , \bom _{\rm{new}} , \bde _{\rm{old}} , \bom _{\rm{old}} ) \ g( \bde _{\rm{new}} , \bom _{\rm{new}} ; \la ; \O ^{(p)} ) \ g( \bde _{\rm{old}} , \bom _{\rm{old}} ; \la ; \bGa ) \text{,} \non 
\end{align}
where $\Theta = \mathbb{R} ^p \times (- \pi / 2, \pi / 2)^{p (p - 1) / 2}$. %
The above expression is simplified by using
\begin{align}
&\sqrt{\la } \{ \tr \bGa - \la \tr ( \bGa [ \{ \bOm ( \bom ) \} ^{\top } \bDe ( \bde ) \bOm ( \bom )]^{- 1} ) \} \non \\
&= \tr \Big[ \bGa \{ \bOm ( \bom ) \} ^{\top } \Big( \bdiag \Big( {\xi _1 \over 1 + \xi _1 / \sqrt{\la }}, \dots , {\xi _p \over 1 + \xi _p / \sqrt{\la }} \Big) \Big) \bOm ( \bom ) \Big] \text{.} \non 
\end{align}

Now, by using Lemma~\ref{lem:clt} that we will prove later, for each $i = 1, \dots , p$, we have
\begin{align}
\exp \Big[ - \la \Big\{ {\xi _i \over \sqrt{\la }} - \log \Big( 1 + {\xi _i \over \sqrt{\la }} \Big) \Big\} \Big] &\le \exp \Big( - {1 \over 2} {{\xi _i}^2 \over 1 + | \xi _i |} \Big) \non 
\end{align}
for all $\xi _i > - \sqrt{\la }$, and 
\begin{align}
\lim_{\la \to \infty } \exp \Big[ - \la \Big\{ {\xi _i \over \sqrt{\la }} - \log \Big( 1 + {\xi _i \over \sqrt{\la }} \Big) \Big\} \Big] = \exp \Big( - {1 \over 2} {\xi _i}^2 \Big) \non 
\end{align}
for all $\xi _i \in \mathbb{R}$. 
Therefore, 
\begin{align}
\lim_{\la \to \infty } g( \bxi , \bom ; \la ; \bGa ) &= g( \bxi , \bom ; \infty ) \non \\
&= J( \bom ) \Big\{ \prod_{1 \le i < j \le p} ( \xi _i - \xi _j ) \Big\} \Big\{ \prod_{i = 1}^{p} \exp \Big( - {1 \over 2} {\xi _i}^2 \Big) \Big\} \mathbbm{1}( \xi _1 > \dots > \xi _p ) \non 
\end{align}
at each $( \bxi , \bom ) \in \Th $. Hence, the limiting function, $g( \bxi , \bom ; \infty )$, is integrable and non-negative, and becomes a probability density after normalization. 
Similarly, we have
\begin{align}
&g( \bxi , \bom ; \la ; \bGa ) \le \Big\{ \prod_{1 \le i < j \le p} (| \xi _i | + | \xi _j |) \Big\} \prod_{i = 1}^{p} \exp \Big( - {1 \over 2} {{\xi _i}^2 \over 1 + | \xi _i |} \Big) \non 
\end{align}
for all $( \bxi , \bom ) \in \Theta$ for all $\la > 0$, the right hand side of which is integrable and independent of $\lambda$. 
Thus, it follows from the dominated convergence theorem that 
\begin{align}
&\lim_{\la \to \infty } {\rm{AAR}} \non \\
&= 2 \int_{\Theta^2} \Big( \mathbbm{1}( \tr [ \bGa \{ \bOm ( \bom _{\rm{new}} ) \} ^{\top } ( \bdiag \bxi _{\rm{new}} ) \bOm ( \bom _{\rm{new}} )] > \tr [ \bGa \{ \bOm ( \bom _{\rm{old}} ) \} ^{\top } ( \bdiag \bxi _{\rm{old}} ) \bOm ( \bom _{\rm{old}} )]) \non \\
&\qquad \times {g( \bxi _{\rm{new}} , \bom _{\rm{new}} ; \infty ) \over \int_{\Theta} g( \bxi , \bom ; \infty ) d( \bxi , \bom ) } {g( \bxi _{\rm{old}} , \bom _{\rm{old}} ; \infty ) \over \int_{\Theta} g( \bxi , \bom ; \infty ) d( \bxi , \bom )} \Big) d( \bxi _{\rm{new}} , \bom _{\rm{new}} , \bxi _{\rm{old}} , \bom _{\rm{old}} ) \text{.} \non 
\end{align}
Since the integrand above is symmetric as a function of $( \bxi _{\rm{new}} , \bom _{\rm{new}} )$ and $( \bxi _{\rm{old}} , \bom _{\rm{old}} )$, we conclude that 
\begin{equation*}
   \lim_{\la \to \infty } \mathbb{P}[ \tr ( \bGa {\bSi _{\rm{new}}}^{- 1} ) \le \tr ( \bGa {\bSi _{\rm{old}}}^{- 1} ) ] = \lim_{\la \to \infty } \mathbb{P}[ \tr ( \bGa {\bSi _{\rm{new}}}^{- 1} ) \ge \tr ( \bGa {\bSi _{\rm{old}}}^{- 1} ) ] = 1/2,
\end{equation*}%
and $\lim_{\la \to \infty } {\rm{AAR}} = \lim_{\la \to \infty } 2\mathbb{P}[ \tr ( \bGa {\bSi _{\rm{new}}}^{- 1} ) \le \tr ( \bGa {\bSi _{\rm{old}}}^{- 1} ) ] = 1$. This completes the proof. 
\end{proof}

\begin{lem}
\label{lem:clt} 
\hfill
\begin{itemize}
\item[{\rm{(i)}}]
For any $\la \ge 1$, we have that 
\begin{align}
&\la \Big\{ {\xi \over \sqrt{\la }} - \log \Big( 1 + {\xi \over \sqrt{\la }} \Big) \Big\} \ge \begin{cases} \displaystyle {1 \over 2} {{\xi }^2 \over 1 + | \xi |} \text{,} & \text{if $\xi > 0$} \text{,} \\ \displaystyle {1 \over 2} {\xi }^2 \text{,} & \text{if $\xi < 0$} \text{,} \non \end{cases} 
\end{align}
all $\xi > - \sqrt{\la }$. 
\item[{\rm{(ii)}}]
For all $\xi \in \mathbb{R}$, we have 
\begin{align}
\lim_{\la \to \infty } \la \Big\{ {\xi \over \sqrt{\la }} - \log \Big( 1 + {\xi \over \sqrt{\la }} \Big) \Big\} = {1 \over 2} \xi ^2 \text{.} \non 
\end{align}
\end{itemize}
\end{lem}

\begin{proof}
For part (i), let $\th = \xi / \sqrt{\la }$. 
Suppose first that $\xi > 0$. 
Then $\th > 0$ and 
\begin{align}
\log \Big( 1 + {\xi \over \sqrt{\la }} \Big) &= - \log \Big( 1 - {\th \over 1 + \th } \Big) = \sum_{k = 1}^{\infty } {1 \over k} \Big( {\th \over 1 + \th } \Big) ^k \text{.} \label{lcltp1} 
\end{align}
Therefore, 
\begin{align}
{\xi \over \sqrt{\la }} - \log \Big( 1 + {\xi \over \sqrt{\la }} \Big) &\ge \th - {\th \over 1 + \th } - \sum_{k = 2}^{\infty } {1 \over 2} \Big( {\th \over 1 + \th } \Big) ^k = {1 \over 2} {\th ^2 \over 1 + \th } \ge {1 \over 2} {\xi ^2 / \la \over 1 + \xi } \text{,} \non 
\end{align}
which implies that 
\begin{align}
\la \Big\{ {\xi \over \sqrt{\la }} - \log \Big( 1 + {\xi \over \sqrt{\la }} \Big) \Big\} \ge {1 \over 2} {{\xi }^2 \over 1 + | \xi |} \text{.} \non 
\end{align}
Next, suppose that $-\sqrt{\la}<\xi < 0$. 
Then $- 1 < \th < 0$ and 
\begin{align}
- \log \Big( 1 + {\xi \over \sqrt{\la }} \Big) &= - \log (1 - | \th |) = \sum_{k = 1}^{\infty } {1 \over k} | \th |^k \text{.} \label{lcltp2} 
\end{align}
Therefore, 
\begin{align}
{\xi \over \sqrt{\la }} - \log \Big( 1 + {\xi \over \sqrt{\la }} \Big) &= \sum _{k=2}^{\infty}\frac{1}{k}|\theta|^k \ge {1 \over 2} \th ^2 = {1 \over 2} {\xi ^2 \over \la} \text{.} \non 
\end{align}

For part (ii), suppose first that $\xi > 0$. 
Then, by (\ref{lcltp1}), 
\begin{align}
\la \Big\{ {\xi \over \sqrt{\la }} - \log \Big( 1 + {\xi \over \sqrt{\la }} \Big) \Big\} &= \la \Big\{ {( \xi / \sqrt{\la } )^2 \over 1 + \xi / \sqrt{\la }} - {1 \over 2} {( \xi / \sqrt{\la } )^2 \over (1 + \xi / \sqrt{\la } )^2} - \sum_{k = 3}^{\infty } {1 \over k} \Big( {\xi / \sqrt{\la } \over 1 + \xi / \sqrt{\la }} \Big) ^k \Big\} \non \\
&= {\xi ^2 \over 1 + \xi / \sqrt{\la }} - {1 \over 2} {\xi ^2 \over (1 + \xi / \sqrt{\la } )^2} - \xi ^2 \sum_{k = 1}^{\infty } {1 \over k + 2} {1 \over (1 + \xi / \sqrt{\la } )^2} \Big( {\xi / \sqrt{\la } \over 1 + \xi / \sqrt{\la }} \Big) ^k \non 
\end{align}
for all $\la > 0$. 
Since 
\begin{align}
{1 \over k + 2} {1 \over (1 + \xi / \sqrt{\la } )^2} \Big( {\xi / \sqrt{\la } \over 1 + \xi / \sqrt{\la }} \Big) ^k &\le \Big( {\xi \over 1 + \xi } \Big) ^k \non 
\end{align}
for all $\la > 1$ for all $k \ge 1$, it follows from the dominated convergence theorem that 
\begin{align}
\lim_{\la \to \infty } \la \Big\{ {\xi \over \sqrt{\la }} - \log \Big( 1 + {\xi \over \sqrt{\la }} \Big) \Big\} &= {1 \over 2} \xi ^2 \text{.} \non 
\end{align}
Next, suppose that $\xi < 0$. 
Then, by (\ref{lcltp2}), 
\begin{align}
\la \Big\{ {\xi \over \sqrt{\la }} - \log \Big( 1 + {\xi \over \sqrt{\la }} \Big) \Big\} &= \la \Big\{ {1 \over 2} {\xi ^2 \over \la } + \sum_{k = 3}^{\infty } {1 \over k} \Big( {| \xi | \over \sqrt{\la }} \Big) ^k \Big\} \non 
\end{align}
for all $\la > \xi ^2$. 
By the dominated convergence theorem, the right-hand side of the above equality converges to $\xi ^2 / 2$ as $\la \to \infty $. 
\end{proof}

\subsection{Example~2: Large $\bPsi$}

\begin{prp}
\label{prp:aar_psi} 
Suppose that $\la \ge - 1$, $\bPsi / 2 = \bdiag ( \psi , 1, \dots , 1) > \O ^{(p)}$, and $\bGa / 2 = \I ^{(p)}$. 
Then, as $\psi \to \infty $, the average acceptance rate converges to $0$. 
\end{prp}

\begin{proof}
First, we have 
\begin{align}
{\rm{AAR}} ( \la , \bPsi , \bGa ) &= 2 \mathbb{E}[ \ \mathbbm{1}[ \tr ( {\bSi _{\rm{new}}}^{- 1} ) \le \tr ( {\bSi _{\rm{old}}}^{- 1} ) ] \ ] \non \\
&= 2 \mathbb{E}[ \ \mathbbm{1}[ \tr ( \{ \bdiag ( \psi , 1, \dots , 1) \} {\bSit _{\rm{new}}}^{- 1} ) \le \tr ( {\bSi _{\rm{old}}}^{- 1} ) ] \ ] \text{,} \non 
\end{align}
where 
\begin{align}
&\bSit _{\rm{new}} = ( \bPsi / 2)^{1 / 2} \bSi ( \bPsi / 2)^{1 / 2} \text{.} \non 
\end{align}
The density of $\bSit _{\rm{new}}$ is proportional to $| \bSit _{\rm{new}} |^{\la } \etr ( - \bSit _{\rm{new}} )$, which is independent of $\psi $. 

Next, we consider the change-of-variables for $\bSi _{\rm{old}}$ as follows. Let $\a = ( a_i )_{i = 1}^{p} \in (0, \infty )^p$ and $\b = (( b_{i, j} )_{j = 1}^{i - 1} )_{i = 2}^{p} \in \mathbb{R} ^{p (p - 1) / 2}$ be such that $\bSi _{\rm{old}} = \B \A \B ^{\top }$ for $\A = \bdiag ( a_1 , \dots , a_p )$ and  %
\begin{align}
\B = \begin{pmatrix} \begin{pmatrix} 1 \\ \b _1 \end{pmatrix} & \begin{pmatrix} \e _{2}^{(2)} \\ \b _2 \end{pmatrix} & \cdots & \begin{pmatrix} \e _{p - 1}^{(p - 1)} \\ \b _{p - 1} \end{pmatrix} & \e _{p}^{(p)} \end{pmatrix} = \begin{pmatrix} 1 & 0 & \cdots & 0 & 0 \\ b_{2, 1} & 1 & \cdots & 0 & 0 \\ \vdots & \vdots & \ddots & \vdots & \vdots \\ b_{p - 1, 1} & b_{p - 1, 2} & \cdots & 1 & 0 \\ b_{p, 1} & b_{p, 2} & \cdots & b_{p, p - 1} & 1 \end{pmatrix} \text{.} \non 
\end{align}
Note that the trace in the MGIG density is written as 
\begin{align}
\tr ( \A \B ^{\top} \bPsi \B ) &= \sum_{i = 1}^{p} a_i \{ ( \e _{i}^{(p)} )^{\top } \B ^{\top} \bPsi \B \e _{i}^{(p)} \} \non \\
&= \sum_{i = 1}^{p} a_i (( \bm{0} ^{(i - 1)} )^{\top } , 1, {\b _i}^{\top } ) \begin{pmatrix} \psi & 0 & \cdots & 0 \\ 0 & 1 & \cdots & 0 \\ \vdots & \vdots & \ddots & \vdots \\ 0 & 0 & \cdots & 1 \end{pmatrix} \begin{pmatrix} \bm{0} ^{(i - 1)} \\ 1 \\ \b _i \end{pmatrix} \non \\
&= a_1 ( \psi + \| \b _1 \| ^2 ) + \sum_{i = 2}^{p} a_i (1 + \| \b _i \| ^2 ) \text{.} \non 
\end{align}
Note also that we have $( \B ^{- 1} ( \B ^{- 1} )^{\top } )_{1, 1} = 1$. Then, the density of $(\a,\b)$ is written as 
\begin{align}
&p( \a , \b ; \psi , \la ) \non \\
&\propto {a_1}^{\la + p - 1} \exp \{ - a_1 ( \psi + \| \b _1 \| ^2 ) \} \exp (- 1 / a_1 ) \non \\
&\quad \times \Big( \prod_{i = 2}^{p} {a_i}^{\la + p - i} \Big) \exp \Big\{ - \sum_{i = 2}^{p} a_i (1 + \| \b _i \| ^2 ) \Big\} \etr \Big\{ - \bdiag \Big( {1 \over a_2} , \dots , {1 \over a_p} \Big) ( \B ^{- 1} ( \B ^{- 1} )^{\top } )_{2:p, 2:p} \Big\} \text{.} \non 
\end{align}
In the expression above, the density kernel depends on $\psi$ %
via $a_1 ( \psi + \| \b _1 \| ^2 )$. 
We transform $(a_1,\a_{2:p},\b)$ to $(\theta,\a_{2:p},\b)$ by: $a_1 = \sqrt{1 / ( \psi + \| \b _1 \| ^2 )} \al $, $\alt = \sqrt{\al }$, $\xi = \alt - 1 / \alt $, and $\xi = \th / ( \psi + \| \b _1 \| ^2 )^{1 / 4}$. That is, $a_1$ is written as 
\begin{align}
a_1 = {1 \over \sqrt{\psi + \| \b _1 \| ^2}} \Big[ {\th / ( \psi + \| \b _1 \| ^2 )^{1 / 4} + \sqrt{\{ \th / ( \psi + \| \b _1 \| ^2 )^{1 / 4} \} ^2 + 4} \over 2} \Big] ^2 \text{.} \non 
\end{align}
Using this expression, we can rewrite $\tr ( {\bSi _{\rm{old}}}^{- 1} )$ as 
\begin{align}
\tr ( {\bSi _{\rm{old}}}^{- 1} ) &= 2 \sqrt{\psi + \| \b _1 \| ^2} / \Big\{ {\th ^2 \over \sqrt{\psi + \| \b _1 \| ^2}} + 2 + {\th \over ( \psi + \| \b _1 \| ^2 )^{1 / 4}} \sqrt{{\th ^2 \over \sqrt{\psi + \| \b _1 \| ^2}} + 4} \Big\} \non \\
&\quad + \sum_{i = 2}^{p} {1 \over a_i} ( \B ^{- 1} ( \B ^{- 1} )^{\top } )_{i, i} \non 
\end{align}
and therefore $\tr ( {\bSi _{\rm{old}}}^{- 1} ) \sim \sqrt{\psi }$ as $\psi \to \infty $. This shows that, for any value of $\tilde{\bSi}_{\rm{new}}$ and $(\theta ,\a_{2:p},\b)$, the indicator function, $\mathbbm{1}[ \tr \{ ( \bdiag ( \psi , 1, \dots , 1) ) {\bSit _{\rm{new}}}^{- 1} \} \le \tr ( {\bSi _{\rm{old}}}^{- 1} ) ]$, converges to zero as $\psi \to \infty$. Below, we show that the density of $(\theta ,\a_{2:p},\b)$ is bounded by an integrable, non-negative function that is independent of $\psi$. Since the density of $\tilde{\bSi}_{\rm{new}}$ does not involve $\psi$, by the dominated convergence theorem, we conclude that the AAR converges to zero.

To study the density of $(\theta ,a_{2:p},\b)$, we define and evaluate its density kernel $g$ as follows: 
\begin{align}
&p( \th , \a _{2:p} , \b ; \psi ; \la ) \propto g( \th , \a _{2:p} , \b ; \psi ; \la ) \non \\
&= {1 \over (1 + \| \b _1 \| ^2 / \psi )^{( \la + p) / 2 + 1 / 4}} \Big\{ 2 {\th \over ( \psi + \| \b _1 \| ^2 )^{1 / 4}} + \sqrt{4 + {\th ^2 \over ( \psi + \| \b _1 \| ^2 )^{1 / 2}}} + {\th ^2 / \sqrt{\psi + \| \b _1 \| } \over \sqrt{4 + \th ^2 / \sqrt{\psi + \| \b _1 \| ^2}}} \Big\} \non \\
&\quad \times \Big\{ {\th ^2 \over \sqrt{\psi + \| \b _1 \| ^2}} + 2 + {\th \over ( \psi + \| \b _1 \| ^2 )^{1 / 4}} \sqrt{4 + {\th ^2 \over \sqrt{\psi + \| \b _1 \| ^2}}} \Big\} ^{\la + p - 1} \exp \Big( - \th ^2 - 2 {\| \b _1 \| ^2 \over \sqrt{\psi + \| \b _1 \| ^2} + \sqrt{\psi }} \Big) \non \\
&\quad \times \Big( \prod_{i = 2}^{p} {a_i}^{\la + p - i} \Big) \exp \Big\{ - \sum_{i = 2}^{p} a_i (1 + \| \b _i \| ^2 ) \Big\} \etr \Big\{ - \bdiag \Big( {1 \over a_2}, \dots , {1 \over a_p} \Big) ( \B ^{- 1} ( \B ^{- 1} )^{\top } )_{2:p, 2:p} \Big\} \non \\
&\le (2 | \th | + \sqrt{4 + \th ^2 } + \th ^2 / 2) ( \th ^2 + 2 + | \th | \sqrt{4 + \th ^2 } )^{\la + p - 1} \exp (- \th ^2 ) \non \\
&\quad \times \Big( \prod_{i = 2}^{p} {a_i}^{\la + p - i} \Big) \exp \Big\{ - \sum_{i = 2}^{p} a_i (1 + \| \b _i \| ^2 ) \Big\} \etr \Big\{ - \bdiag \Big( {1 \over a_2}, \dots , {1 \over a_p} \Big) ( \B ^{- 1} ( \B ^{- 1} )^{\top } )_{2:p, 2:p} \Big\} \text{,} \non 
\end{align}
where the inequality holds if $\psi \ge 1$. 
The upper bound of $g$ obtained here is clearly independent of $\psi $ and integrable since%
\begin{align}
&\int_{(0, \infty )^{p - 1} \times \mathbb{R} ^{p (p - 1) / 2}} \Big[ \Big( \prod_{i = 2}^{p} {a_i}^{\la + p - i} \Big) \exp \Big\{ - \sum_{i = 2}^{p} a_i (1 + \| \b _i \| ^2 ) \Big\} \non \\
&\quad \times \etr \Big\{ - \Big( \bdiag \Big( {1 \over a_2}, \dots , {1 \over a_p} \Big) \Big) ( \B ^{- 1} ( \B ^{- 1} )^{\top } )_{2:p, 2:p} \Big\} \Big] d( \a _{2:p} , \b ) \non \\
&= \int_{(0, \infty )^{p - 1} \times \mathbb{R} ^{p (p - 1) / 2}} \Big( \Big( \prod_{i = 2}^{p} {a_i}^{\la + p - i} \Big) \exp \Big\{ - \sum_{i = 2}^{p} a_i (1 + \| \b _i \| ^2 ) \Big\} \non \\
&\quad \times \exp \Big\{ - (- \b _1 )^{\top } {\Bbt _2}^{\top } \dotsm {\Bbt _p}^{\top } \Big( \bdiag \Big( {1 \over a_2}, \dots , {1 \over a_p} \Big) \Big) \Bbt _p \dotsm \Bbt _2 (- \b _1 ) \Big\} \non \\
&\quad \times \exp \Big[ - \tr \Big\{ {\Bbt _2}^{\top } \dotsm {\Bbt _p}^{\top } \Big( \bdiag \Big( {1 \over a_2}, \dots , {1 \over a_p} \Big) \Big) \Bbt _p \dotsm \Bbt _2 \Big\} \Big] \Big) d( \a _{2:p} , \b ) \non \\
&\propto \int_{(0, \infty )^{p - 1} \times \mathbb{R} ^{(p - 1) (p - 2) / 2}} \Big[ \Big( \prod_{i = 2}^{p} {a_i}^{\la + p - i + 1 / 2} \Big) \exp \Big\{ - \sum_{i = 2}^{p} a_i (1 + \| \b _i \| ^2 ) \Big\} \non \\
&\quad \times \etr \Big\{ - {\Bbt _2}^{\top } \dotsm {\Bbt _p}^{\top } \Big( \bdiag \Big( {1 \over a_2}, \dots , {1 \over a_p} \Big) \Big) \Bbt _p \dotsm \Bbt _2 \Big\} \Big] d( \a _{2:p} , \b _{- 1} ) \non \\
&\propto \int_{(0, \infty )^{p - 1} \times \mathbb{R} ^{(p - 1) (p - 2) / 2}} {\rm{MGIG}}_{p - 1} ( \tilde{\a } , \tilde{\b } | \la + 1 / 2, 2 \I ^{(p - 1)} , 2 \I ^{(p - 1)} ) d( \a _{2:p} , \b _{- 1} ) < \infty \text{,} \non 
\end{align}
where $\b _{- 1} = \b \setminus \b _1$ and 
\begin{align}
\Bbt _i = \begin{pmatrix} \e _{1}^{(p)} & \cdots & \e _{i - 1}^{(p)} & \begin{pmatrix} \e _{i}^{(i)} \\ - \b _i \end{pmatrix} & \e _{i + 1}^{(p)} & \cdots & \e _{p}^{(p)} \end{pmatrix} _{2:p, 2:p} \text{,} \quad i = 2, \dots , p \text{.} \non 
\end{align}
Also, the limit of the density kernel is 
\begin{align*}
    &\lim _{\psi \to \infty}  g( \th , \a _{2:p} , \b ; \psi ; \la ) = 2^{\la + p} e^{-\theta^2} \\
    &\times \Big( \prod_{i = 2}^{p} {a_i}^{\la + p - i} \Big) \exp \Big\{ - \sum_{i = 2}^{p} a_i (1 + \| \b _i \| ^2 ) \Big\} \etr \Big\{ - \bdiag \Big( {1 \over a_2}, \dots , {1 \over a_p} \Big) ( \B ^{- 1} ( \B ^{- 1} )^{\top } )_{2:p, 2:p} \Big\} ,
\end{align*}
which is also integrable. Hence, the normalizing constant of $g$ is shown to converge to some finite, non-zero value as $\psi \to \infty$. This shows that the original density, $p( \th , \a _{2:p} , \b ; \psi ; \la )$, is bounded by an integrable function that is independent of $\psi$. This completes the proof. 

\end{proof}

\section{Acceralation of the Gibbs sampler by parallellization}
For $i = 1, \dots , p$, let $\P ^{(i)} = ( \e _{i}^{(i)} , \dots , \e _{1}^{(i)} )$ and note that $\P ^{(i)} = ( \P ^{(i)} )^{\top } = ( \P ^{(i)} )^{- 1}$. 
Let $\bPsit ^{1 / 2}$ be the lower triangular matrix with positive diagonal elements satisfying $\bPsit ^{1 / 2} ( \bPsit ^{1 / 2} )^{\top } = \P ^{(p)} \bPsi \P ^{(p)}$ and write $\bPsit ^{- 1 / 2} = ( \bPsit ^{1 / 2} )^{- 1}$. 

\begin{lem}
\label{lem:parallelization} 
Let $\U _i$ and $\bLa _i$ be orthogonal and diagonal matrices such that $\U _i \bLa _i {\U _i}^{\top } = ( \S )_{1:(p - i), 1:(p - i)}$ for $i = 1, \dots , p - 1$, where $\S = \bPsit ^{- 1 / 2} \P ^{(p)} ( \B ^{- 1} )^{\top } \A ^{- 1} \B ^{- 1} \P ^{(p)} ( \bPsit ^{- 1 / 2} )^{\top }$. 
Then, for all $i = 1, \dots , p - 1$ and all $\al , \mu > 0$, we have 
\begin{align}
&\al ( \bPsi )_{(i + 1):p, (i + 1):p} + \mu (( \B ^{- 1} )^{\top } \A ^{- 1} \B ^{- 1} )_{(i + 1):p, (i + 1):p} \non \\
&= \P ^{(p - i)} ( \bPsit ^{1 / 2} )_{1:(p - i), 1:(p - i)} \U _i ( \al \I ^{(p - i)} + \mu \bLa _i ) {\U _i}^{\top } \{ ( \bPsit ^{1 / 2} )_{1:(p - i), 1:(p - i)} \} ^{\top } \P ^{(p - i)} \text{.} \non 
\end{align}
\end{lem}

\begin{proof}
We have 
\begin{align}
&\al ( \bPsi )_{(i + 1):p, (i + 1):p} + \mu (( \B ^{- 1} )^{\top } \A ^{- 1} \B ^{- 1} )_{(i + 1):p, (i + 1):p} \non \\
&= \P ^{(p - i)} ( \P ^{(p)} \{ \al \bPsi + \mu ( \B ^{- 1} )^{\top } \A ^{- 1} \B ^{- 1} \} \P ^{(p)} )_{1:(p - i), 1:(p - i)} \P ^{(p - i)} \non \\
&= \P ^{(p - i)} ( \bPsit ^{1 / 2} ( \al \I ^{(p)} + \mu \S ) ( \bPsit ^{1 / 2} )^{\top } )_{1:(p - i), 1:(p - i)} \P ^{(p - i)} \non \\
&= \P ^{(p - i)} ( \bPsit ^{1 / 2} )_{1:(p - i), 1:(p - i)} ( \al \I ^{(p - i)} + \mu \U _i \bLa _i {\U _i}^{\top } ) (( \bPsit ^{1 / 2} )^{\top } )_{1:(p - i), 1:(p - i)} \P ^{(p - i)} \non \\
&= \P ^{(p - i)} ( \bPsit ^{1 / 2} )_{1:(p - i), 1:(p - i)} \U _i ( \al \I ^{(p - i)} + \mu \bLa _i ) {\U _i}^{\top } \{ ( \bPsit ^{1 / 2} )_{1:(p - i), 1:(p - i)} \} ^{\top } \P ^{(p - i)} \text{.} \non 
\end{align}
\end{proof}

It follows from Lemma \ref{lem:parallelization} that we can easily update $\b _i$, $i = 1, \dots, p - 1$, after first decomposing $( \bPsit ^{- 1 / 2} \P ^{(p)} ( \B ^{- 1} )^{\top } \A ^{- 1} \B ^{- 1} \P ^{(p)} ( \bPsit ^{- 1 / 2} )^{\top } )_{1:i, 1:i}$ for each $i = 1, \dots , p$, for which we could use parallelization. 
Note that the approach here is to compute eigenpairs instead of inverses.

\section{Additional results on the simulation study in Section~4.1}

In Section~4.1, we studied the computational efficiencies of the Gibbs sampler and MH methods in the numerical experiment when the order parameter of the MGIG distribution is set to $\la = 2$. We changed this value to $\la = 10$ and conducted the same experiment. The ESSs and ESSs per second in this experiment are summarized in Figure~S1. The performance of the MH methods improve, which is consistent with the results reported in the literature. We would like to emphasize that the success of the MH methods for large $\la$ is not guaranteed in more complex statistical models, as evidenced in our example of the partial Gaussian graphical models in Section~4.2.

\begin{figure}[!htb]
\centering
\includegraphics[width = \linewidth]{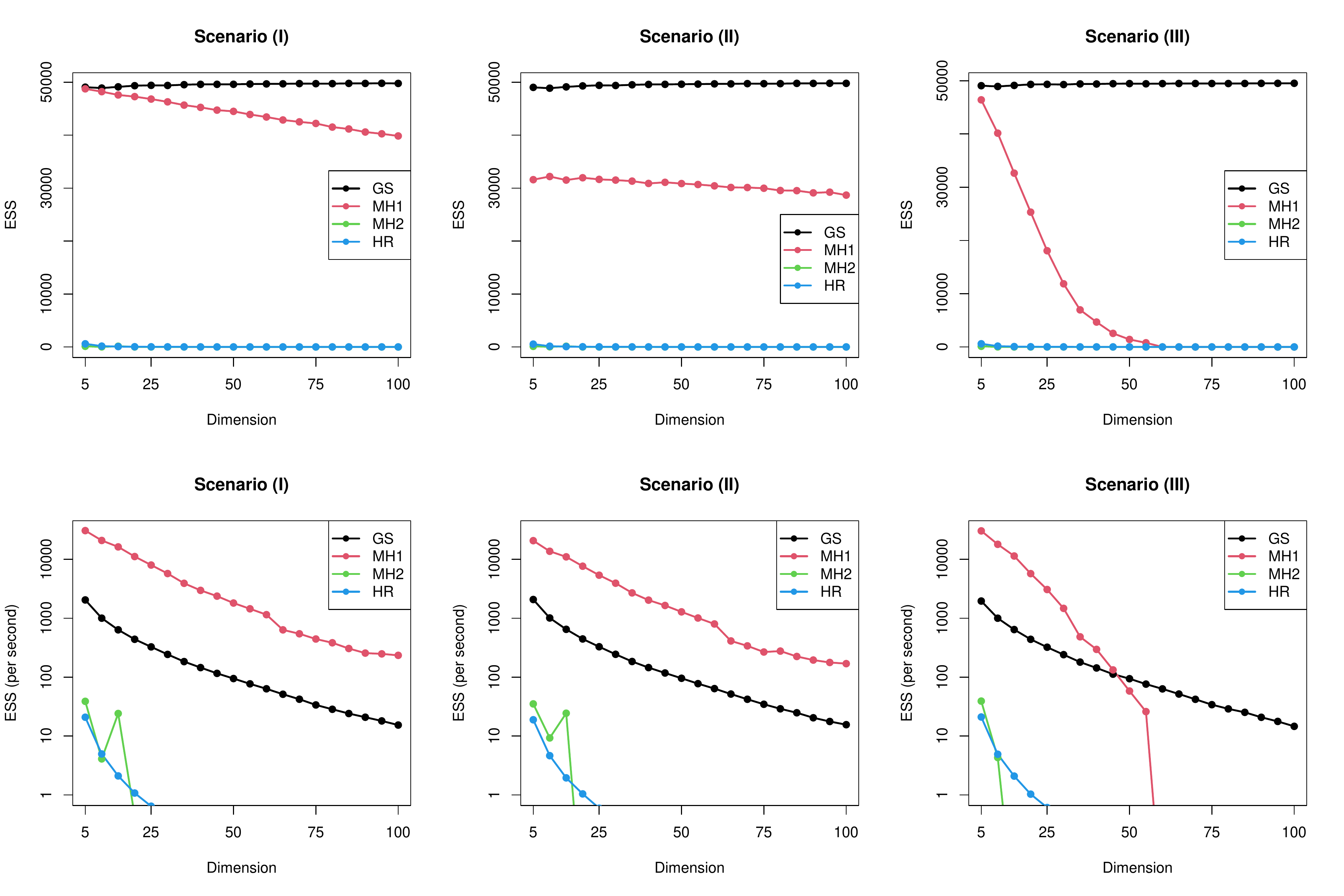}
\caption{Effective sample size (ESS) and ESS per second of the four samplers when $\la = 10$.}
\label{fig:generation10}
\end{figure}

\section{Detailed MCMC algorithm for the matrix skew-$t$ distribution in Section~\ref{sec:num-MST}}

The full conditional distributions of $\W_i$, $\M$, $\B$ and $\bPsi$ are as follows:

\begin{itemize}
\item[-] 
The full conditional distribution of the latent matrix $\W_i$ is proportional to 
$$
|\W_i|^{-(\nu+p+q+1)/2}\exp\left\{-\frac12{\rm tr}(\tilde{\bPhi}_i\W_i+\tilde{\bGa}_i\W_i^{-1})\right\},
$$
where 
$$
\tilde{\bPhi}_i=\B\bOm^{-1}\B^\top, \ \ \ \ 
\tilde{\bGa}_i=\bPsi + (\Y_i-\M)\bOm^{-1}(\Y_i-\M)^\top.
$$
Note that ${\rm rank}(\tilde{\bPhi}_i)=\min(p,q)$ and ${\rm rank}(\tilde{\bGa}_i)=p$ when $\bPsi$ is positive definite. 
Hence, the full conditional of $\W_i$ is ${\rm MGIG}_p(-(\nu+p+q+1)/2, \tilde{\bPhi}_i, \tilde{\bGa}_i)$.

\item[-]
The full conditional of ${\rm vec}(\M)$ is $N_{pq}( \Dbt_M\dbt_M, \Dbt_M)$, where 
\begin{align*}
&\Dbt_M=\left\{\bOm^{-1}\otimes \left(\sum_{i=1}^n\W_i^{-1}\right) + \V_{0M}^{-1}\otimes \U_{0M}^{-1}\right\}^{-1},\\
&\dbt_M=\sum_{i=1}^n(\bOm^{-1}\otimes \W_i^{-1}){\rm vec}(\Y_i-\W_i\B)+(\V_{0M}^{-1}\otimes \U_{0M}^{-1}){\rm vec}(\A_{0M}).
\end{align*}

\item[-]
The full conditional of ${\rm vec}(\B)$ is $N_{pq}(\Dbt_B\dbt_B, \Dbt_B)$, where 
\begin{align*}
&\Dbt_B=\left\{\bOm^{-1}\otimes \left(\sum_{i=1}^n\W_i\right) + \V_{0B}^{-1}\otimes \U_{0B}^{-1}\right\}^{-1},\\
&\dbt_B=\sum_{i=1}^n(\bOm^{-1}\otimes \I_p){\rm vec}(\Y_i-\M)+(\V_{0B}^{-1}\otimes \U_{0B}^{-1}){\rm vec}(\A_{0B}).
\end{align*}

\item[-]
The full conditional of $\bPsi$ is ${\rm W}_p((\sum_{i=1}^n\W_i^{-1}+\bPsi_0^{-1})^{-1}, \eta_0+n\nu)$.

\item[-]
The full conditional of $\bOm$ is ${\rm IW}_p(\bOm_0+\sum_{i=1}^n(\Y_i-\M-\W_i\B)^\top\W_i^{-1}(\Y_i-\M-\W_i\B), \xi_0+np)$.
\end{itemize}

\end{document}